\documentclass{eptcs}
\providecommand{\event}{GandALF 2024} 

\usepackage{iftex}

\ifpdf
  \usepackage{underscore}         
  \usepackage[T1]{fontenc}        
\else
  \usepackage{breakurl}           
\fi

\title{Deterministic Suffix-reading Automata}

\author{R Keerthan
\institute{Tata Consultancy Services Innovation Labs \\ Pune, India}
\institute{Chennai Mathematical Institute, India}
\email{keerthan.r@tcs.com}
\and
B Srivathsan
\institute{Chennai Mathematical Institute, India}
\institute{CNRS IRL 2000, ReLaX, Chennai, India}
\email{sri@cmi.ac.in}
\and
R Venkatesh \qquad\qquad Sagar Verma
\institute{Tata Consultancy Services Innovation Labs \\
Pune, India}
\email{\quad r.venky@tcs.com \quad\qquad verma.sagar2@tcs.com
  \footnote{All authors have contributed equally and are listed in the alphabetical order of last names.}}
}
\def\titlerunning{Deterministic Suffix-reading Automata}
\def\authorrunning{R. Keerthan, B. Srivathsan, R. Venkatesh \& S. Verma}

\usepackage{amsmath,amssymb, amsthm, thm-restate}
\usepackage{stmaryrd}
\usepackage{tikz}
\usepackage{url}
\usepackage{caption}
\usepackage{adjustbox}

\newtheorem{definition}{Definition}

\newtheorem{lemma}{Lemma}

\newtheorem{theorem}{Theorem}

\renewcommand{\a}{\alpha}
\renewcommand{\b}{\beta}
\newcommand{\s}{\sigma}

\newcommand{\e}{\epsilon}

\newcommand{\rr}{0}
\newcommand{\bb}{1}
\renewcommand{\gg}{2} 
\newcommand{\yy}{3} 
\newcommand{\ww}{4} 
\newcommand{\Aa}{\mathcal{A}}
\newcommand{\Ll}{\mathcal{L}}

\newcommand{\Ss}{\mathcal{S}}

\newcommand{\xra}{\xrightarrow}

\newcommand{\incl}{\subseteq}

\newcommand{\prfx}{\sqsubseteq_{\mathsf {pr}}}

\newcommand{\sfx}{\sqsubseteq_{\mathsf {sf}}}

\newcommand{\NP}{\mathsf{NP}}

\newcommand{\out}{\operatorname{Out}}
\newcommand{\outp}{\overline{\operatorname{Out}}}

\renewcommand{\e}{\varepsilon}
\renewcommand{\epsilon}{\e}

\newcommand{\spath}[3]{\ensuremath{\mathsf{SP}(#1 \rightsquigarrow #2,
#3)}}
\newcommand{\spaths}[2]{\ensuremath{\mathsf{SP}(#1, #2)}}

\begin{document}

\maketitle

\begin{abstract}
  We introduce deterministic suffix-reading automata (DSA), a new
  automaton model over finite words. Transitions in a DSA are labeled
  with words. From a state, a DSA triggers an outgoing transition on seeing
  a word \emph{ending} with the transition's label. Therefore,
  rather than moving along an input word letter by letter, a DSA can
  jump along blocks of letters, with each block ending in a suitable
  suffix. This feature allows DSAs to recognize regular languages more
  concisely, compared to DFAs. In this work, we focus on questions
  around finding a ``minimal'' DSA for a regular language. The number
  of states is not a faithful measure of the size of a DSA, since the
  transition-labels contain strings of arbitrary length. Hence, we
  consider total-size (number of states + number of edges + total
  length of transition-labels) as the size measure of DSAs.

  We start by formally defining the model and providing a DSA-to-DFA
  conversion that allows to compare the expressiveness and
  succinctness of DSA with related automata models.  Our main
  technical contribution is a method to \emph{derive} DSAs from a
  given DFA: a DFA-to-DSA conversion. 
  We make a surprising
  observation that the smallest DSA derived from the canonical DFA of
  a regular language $L$ need not be a minimal DSA for $L$. This
  observation leads to a fundamental bottleneck in deriving a minimal
  DSA for a regular language. In fact, we prove that
  given a DFA and a number $k \ge 0$, the problem of deciding if there
  exists an equivalent DSA of total-size $\le k$ is NP-complete.

\end{abstract}

\section{Introduction}

Deterministic Finite Automata (DFA) are fundamental to many areas in
Computer Science. Apart from being the cornerstone in the study of
regular languages, automata have been applied in several contexts:
such as text processing~\cite{DBLP:journals/tcs/MohriMW09}, 
model-checking~\cite{DBLP:books/daglib/0007403-2}, software
verification~\cite{DBLP:conf/cav/BouajjaniJNT00,
  DBLP:conf/cav/BouajjaniHV04, 989841}, and formal specification
languages~\cite{DBLP:journals/scp/Harel87}.  A central
challenge in the application of automata is the size of the automaton
involved.  Non-determinism gives exponential succinctness, however, a
deterministic model is useful in formal specifications and automata
implementations. The literature offers different ways to get succinct
representations of DFAs.  We recall a few of them below and propose a
new solution to this problem.

One of the reasons for large DFAs is the size of the alphabet, for
instance, consider the alphabet of all ASCII characters. Having a
transition for each letter from each state blows up the size of the
automata. Symbolic automata~\cite{DBLP:conf/popl/VeanesHLMB12,
  DBLP:conf/cav/DAntoniV17} have been proposed to handle large
alphabets. Letters on the edges are replaced by formulas, which club
together several transitions between a pair of states into one
symbolic transition. Symbolic automata have been implemented in many
tools and have been widely applied (see \cite{loris-page} for a list
of tools and applications).

Another dimension in reducing the DFA representation is to consider
transitions on a block of letters. Generalized automata (GA) are
extensions of non-deterministic finite automata (NFAs) that can
contain strings instead of letters on transitions. A word $w$ is
accepted if it can be broken down as $w_1 w_2 \dots w_k$ such that
each segment is read by a transition. This model was defined by
Eilenberg~\cite{DBLP:books/lib/Eilenberg74}, and later
Hashiguchi~\cite{DBLP:conf/icalp/Hashiguchi91} proved that for every
regular language $L$ there is a minimal GA in which the edge labels
are at most a polynomial function in $m$, where $m$ is size of the syntactic monoid
of
$L$.
Giammarresi \emph{et al.}~\cite{giammarresi1999deterministic}
considers deterministic generalized automata (DGA) and proposes an
algorithm to generate a minimal DGA (in terms of the number of states)
in which the edges have length at most the size of the minimal
DFA. The algorithm uses a method to suppress states and create longer
labels. The key observation is that minimal DGAs can be derived from
the canonical DFA by suppressing states.

\textbf{Our model.} In this work, we introduce \emph{Deterministic
  Suffix-reading Automata (DSAs)}.  We continue to work with strings
on transitions, as in DGA. However, the meaning of transitions is
different. A transition $q \xra{abba} q'$ is enabled if at $q$, a word
$w$ \emph{ending} with $abba$ is seen, and moreover no other
transition out of $q$ is enabled at a prefix of $w$. Intuitively, the
automaton tracks a finite set of pattern strings at each state. It
stays in a state until one of them appears as the \emph{suffix} of the
word read so far, and then makes the appropriate transition.
We start
with a motivating example. Consider a model for out-of-context \texttt{else} statements, in
relation to \texttt{if} and \texttt{endif} statements in a programming
language. Assume a suitable alphabet $\Sigma$ of characters.  Let
$L_{\texttt{else}}$ be the set of all strings over the alphabet where
(1) there are no nested \texttt{if} statements, and (2) there is an
\texttt{else} which is not between an \texttt{if} and an
\texttt{endif}.  
A DFA for this language performs string matching to detect the
\texttt{if}, \texttt{else} and \texttt{endif}. The DSA is shown in
Figure~\ref{fig:if-else}: at $s_0$, it passively reads letters until
it first sees an \texttt{if} or an \texttt{else}. If it is an
\texttt{if}, the automaton transitions to $s_1$. For instance, on a
word \texttt{abf4fgif} the automaton goes to $s_1$, since it ends with
\texttt{if} and there is no \texttt{else} seen so far. Similarly, at
$s_1$ it waits for one of the patterns \texttt{if} or an
\texttt{endif}. If it is the former, it goes to $s_3$ and rejects,
otherwise it moves to $s_0$, and so on.

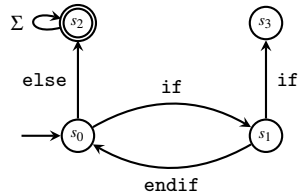
\begin{figure}[t]
  \centering
  \begin{tikzpicture}[state/.style={circle, draw, thick, inner sep =
      2pt}]
    \begin{scope}[every node/.style={state}]
      \node (0) at (0,0) {\tiny $s_0$}; \node (1) at (2.5, 0) {\tiny
        $s_1$}; \node [double] (2) at (0, 1.5) {\tiny $s_2$}; \node
      (3) at (2.5, 1.5) {\tiny $s_3$};
    \end{scope}
    \begin{scope}[->, >=stealth, thick, auto]
      \draw (-0.75, 0) to (0); \draw (0) to node {\scriptsize
        $\mathtt{else}$ } (2); \draw (0) to [bend left=30] node
      {\scriptsize $\mathtt{if}$} (1); \draw (1) to [bend left=30]
      node {\scriptsize $\mathtt{endif}$} (0); \draw (2) to [loop
      left] node {\scriptsize $\Sigma$} (2); \draw (1) to node [right]
      {\scriptsize $\mathtt{if}$} (3);
    \end{scope}
  \end{tikzpicture}
  \caption{DSA for out-of-context \texttt{else}}
  \label{fig:if-else}
\end{figure}

Suffix-reading automata have the ability to wait at a state, reading
long words until a matching pattern is seen. This results
in an arguably more readable specification for languages which are
``pattern-intensive''. This representation is orthogonal to the
approaches considered so far. Symbolic automata club together
transitions between a pair of states, whereas DSA can do this clubbing
across several states and transitions. DGA have this facility of
clubbing across states, but they cannot ignore intermediate letters,
which results in extra states and transitions.

\textbf{Overview of results.}
We formally present deterministic suffix-reading automata and its
semantics, quantify its size in comparison to an equivalent DFA, and
study an algorithm to construct DSAs starting from a DFA. This is in
the same spirit as in DGAs, where smaller DGAs are obtained by suppressing
states. For automata models with strings on transitions, number of
states is not a faithful measure of the size of a DSA. As described
in~\cite{giammarresi1999deterministic}, we consider the total size of
a DSA which includes the number of states, edges, and the sum of label
lengths. The key contributions of this paper are:
\begin{enumerate}
\item Presentation of a definition of a new kind of automaton - DSA~(Section~\ref{sec:new-automaton-model}).

\item Proof that DSAs accept regular languages, and nothing more.
  Every complete DFA can be seen as a DSA. For the converse, we prove
  that for every DSA of size $k$, there is a DFA with size at most
  $2k \cdot (1 + 2 |\Sigma|)$, where $\Sigma$ is the alphabet
  (Lemma~\ref{lem:tracking-dfa-language-equivalent},
  Theorem~\ref{thm:comparing-dsa-with-dfa-dga}).
  This answers the question of how small DSAs can be in comparison to
  DFAs for a certain language : if $n$ is the size of the minimal DFA
  for a language $L$, minimal DSAs for $L$ cannot be smaller than
  $\frac{n}{2 \cdot (1 + 2
    |\Sigma|)}$. 
  When the alphabet is large, one could expect smaller sized DSAs. We
  describe a family of languages $L_n$, with alphabet size $n$, for
  which the minimal DFA has size quadratic in $n$, whereas size of
  DSAs is a linear function of $n$ (Lemma~\ref{lem:dsa-small}).

\item We present a method to derive DSAs out of DFAs, a DFA-to-DSA
  conversion (Section~\ref{sec:suffix-tracking-sets}). In a nutshell,
  the derivation procedure selects subsets of DFA-states, and adds
  transitions labeled with (some of) the acyclic paths between
  them. Our main
  technical contribution lies in identifying sufficient conditions on
  the selected subset of states, so that the derivation procedure
  preserves the language (Theorem~\ref{def:suffix-tracking-set}).

\item We remark that minimal DSAs need not be unique, and make a
  surprising observation: the smallest DSA that we derive from the
  canonical DFA of $L$ need not be a minimal DSA. We find this
  surprising because (1) firstly, our derivation procedure is
  surjective: every DSA (satisfying some natural assumptions) can be
  derived from some corresponding DFA, and in particular, a minimal
  DSA can be derived from some DFA; (2) the observation suggests that
  one may need to start with a bigger DFA in order to derive a minimal
  DSA -- so, starting with a bigger DFA may result in a smaller DSA
  (Section~\ref{sec:minim-some-observ}).

\item Finally, we show that given a DFA and a number $k$, deciding if
  there exists a DSA of size $\le k$ is $\NP$-complete
  (Section~\ref{sec:complexity}).

\end{enumerate}

\emph{Related work.} The closest to our work
is~\cite{giammarresi1999deterministic} which introduces DGAs, and
gives a procedure to derive DGAs from DFAs. The focus however is on
getting DGAs with as few states as possible. The ideas presented in
Section~\ref{sec:minim-some-observ} of our work, also apply for
state-minimality: the same example shows that in order to get fewer
states, one may have to start with a bigger DFA.  This is in sharp
contrast to the DGA setting, where the derivation procedure of
\cite{giammarresi1999deterministic} yields a minimal DGA (in the
number of states) when applied on the canonical DFA. The problem of
deriving DGAs with minimal total-size was left open
in~\cite{giammarresi1999deterministic}, and continues to remain so, to
the best of our knowledge.
Expression automata~\cite{DBLP:conf/wia/HanW04} allow regular
expressions as transition labels.  This model was already considered
in~\cite{DBLP:journals/tc/BrzozowskiM63} to convert automata to
regular expressions. Every DFA can be converted to a two state
expression automaton with a regular expression connecting them.  A
model of deterministic Expression automata (DEA) was proposed
in~\cite{DBLP:conf/wia/HanW04} with restrictions that limit the
expressive power.  An algorithm to convert a DFA to a DEA, by repeated
state elimination, is proposed in~\cite{DBLP:conf/wia/HanW04}. The
resulting DEA is minimal in the number of states.  The issue with
Expression automata is the high expressivity of the transition
condition, that makes states almost irrelevant. On the other hand, DEA
have restrictions that make the model less expressive than DFAs.
Minimization of NFAs was studied
in~\cite{DBLP:journals/siamcomp/JiangR93} and shown to be
hard.Succinctness of models with different features, like alternation,
two-wayness, pebbles, and a notion of concurrency, has been studied
in~\cite{DBLP:journals/tcs/GlobermanH96}.


\section{Preliminaries}
\label{sec:preliminaries}

We fix a finite alphabet $\Sigma$. Following standard convention, we
write $\Sigma^*$ for the set of all words (including $\e$) over
$\Sigma$, and $\Sigma^+ = \Sigma^* \setminus \{ \e\}$. For
$w \in \Sigma^*$, we write $|w|$ for the length of $w$, with $|\e|$
considered to be $0$.  A word $u$ is a \emph{prefix} of word $w$ if
$w = u v$ for some $v \in \Sigma^*$; it is a \emph{proper-prefix} if
$v \in \Sigma^+$. Observe that $\e$ is a prefix of every word. A set
of words $W$ is said to be a \emph{prefix-free set} if no word in $W$
is a prefix of another word in $W$. A word $u$ is a
\emph{suffix} (resp. \emph{proper-suffix}) of $w$ if $w = vu$ for some
$v \in \Sigma^*$ (resp. $v \in \Sigma^+$).

A \emph{Deterministic Finite Automaton (DFA)} $M$ is a tuple
$(Q, \Sigma, q^{init}, \delta, F)$ where $Q$ is a finite set of
states, $q^{init} \in Q$ is the initial state, $F \incl Q$ is a set of
accepting states, and $\delta: Q \times \Sigma \to Q$ is a partial
function describing the transitions. If $\delta$ is complete, the
automaton is said to be a complete DFA. Else, it is called a trim
DFA. The run of DFA $M$ on a word $w = a_1 a_2 \dots a_n$ (where
$a_i \in \Sigma$) is a sequence of transitions
$(q_0, a_1, q_1) (q_1, a_2, q_2) \dots (q_{n-1}, a_n, q_n)$ where each
$(q_i, a_{i+1}, q_{i+1}) \in \delta$ for $0 \le i < n$, and
$q_0 = q^{init}$, the initial state of $M$. The run is accepting if
$q_n \in F$. If the DFA is complete, every word has a unique run. On a
trim DFA, each word either a has unique run, or it has no run. The
language $\Ll(M)$ of DFA $M$, is the set of words for which $M$ has an
accepting run.

We will now recall some useful facts about minimality of DFAs. Here,
by minimality, we mean DFAs with the least number of states. Every
complete DFA $M$ induces an equivalence $\sim_M$ over words:
$u \sim_M v$ if $M$ reaches the same state on reading both $u$ and $v$
from the initial state.  In the case of trim DFAs, this equivalence
can be restricted to set of prefixes of words in $\Ll(M)$. For a
regular language $L$, we have the Nerode equivalence: $u \approx_L v$
if for all $w \in \Sigma^*$, we have $uw \in L$ iff $v w \in L$. By
the well-known Myhill-Nerode theorem (see
\cite{DBLP:books/daglib/0016921} for more details), there is a
canonical DFA $M_L$ with the least number of states for $L$, and
$\sim_{M_L}$ equals the Nerode equivalence $\approx_L$. Furthermore,
every DFA $M$ for $L$ is a \emph{refinement} of $M_L$: $u \sim_M v$
implies $u \sim_{M_L} v$. If two words reach the same state in $M$,
they reach the same state in $M_L$.

  A \emph{Deterministic Generalized Automaton (DGA)}~\cite{giammarresi1999deterministic} $H$ is given by
  $(Q, \Sigma, q^{init}, E, F)$ where $Q, q^{init}, F$ mean the same
  as in DFA, and $E \incl Q \times \Sigma^+ \times Q$ is a finite set
  of edges labeled with words from $\Sigma^+$. For every state $q$,
  the set $\{ \alpha \mid (q, \alpha, q') \in E \}$ is a prefix-free
  set.
A run of DGA $H$ on a word $w$ is a sequence of edges
$(q_0, \alpha_1, q_1) (q_1, \alpha_2, q_2) \dots (q_{n-1}, \a_n, q_n)$
such that $w = \alpha_1 \alpha_2 \dots \a_n$, with $q_0$ being the
initial state. As usual, the run is accepting if $q_n \in F$. Due to
the property of the set of outgoing labels being a prefix-free set,
there is a atmost one run on every word. The language $\Ll(H)$ is the
set of words with an accepting run. Figure~\ref{fig:dfa-dga-examples}
gives examples of DFAs and corresponding DGAs.

\begin{figure}
  \centering
  \begin{tikzpicture}[state/.style={circle, draw, thick, inner sep =
      2pt}, scale=0.8]
    \begin{scope}[every node/.style={state}]
      \node [double] (0) at (0,0) {\scriptsize $q_0$}; \node (1) at
      (2,0) {\scriptsize $q_1$};
    \end{scope}
    \begin{scope}[->,>=stealth, thick, auto]
      \draw (-0.8, 0) to (0); \draw (0) to [bend left=20] node
      {\scriptsize $a$} (1); \draw (1) to [bend left=20] node
      {\scriptsize $b$} (0);
    \end{scope}
    \node [left] at (-0.7, 0) {\scriptsize $DFA \quad M_1:$};

    \begin{scope}[yshift=-2.5cm]
      \begin{scope}[every node/.style={state}]
        \node [double] (0) at (0,0) {\scriptsize $q_0$};
      \end{scope}
      \begin{scope}[->,>=stealth, thick, auto]
        \draw (-0.8, 0) to (0); \draw (0) to [loop right] node
        {\scriptsize $ab$} (0);
      \end{scope}
      \node [left] at (-0.7, 0) {\scriptsize $DGA \quad H_1:$};
    \end{scope}

    \draw (3.2,-3) to (3.2,1);
    
    \begin{scope}[xshift = 6cm]
      \begin{scope}[every node/.style={state}]
        \node (0) at (0,0) {\scriptsize $q_0$}; \node (1) at (1.8,0)
        {\scriptsize $q_1$}; \node (2) at (3.6,0) {\scriptsize $q_2$};
        \node [double] (3) at (5.4,0) {\scriptsize $q_3$};
      \end{scope}

      \begin{scope}[->,>=stealth, thick, auto]
        \draw (-0.8, 0) to (0); \draw (0) to [loop above] node
        {\scriptsize $b$} (0); \draw (0) to node {\scriptsize $a$}
        (1); \draw (1) to node {\scriptsize $a$} (2); \draw (2) to
        [loop above] node {\scriptsize $a$} (2); \draw (2) to node
        {\scriptsize $b$} (3); \draw (3) to [bend left=20] node
        {\scriptsize $a$} (1); \draw (3) to [bend left=35] node
        {\scriptsize $b$} (0);
       
      \end{scope}
      \node [left] at (-0.7, 0) {\scriptsize $DFA \quad M_2:$};

      \begin{scope}[yshift=-2.5cm]
        \begin{scope}[every node/.style={state}]
          \node (0) at (0,0) {\scriptsize $q_0$}; \node (2) at (3.6,0)
          {\scriptsize $q_2$}; \node [double] (3) at (5.4,0)
          {\scriptsize $q_3$};
        \end{scope}
        \begin{scope}[->,>=stealth, thick, auto]
          \draw (-0.8, 0) to (0); \draw (0) to [loop above] node
          {\scriptsize $b$} (0); \draw (0) to node {\scriptsize $aa$}
          (2); \draw (2) to [loop above] node {\scriptsize $a$} (2);
          \draw (2) to node {\scriptsize $b$} (3); \draw (3) to [bend
          left=20] node {\scriptsize $aa$} (2); \draw (3) to [bend
          left=40] node {\scriptsize $b$} (0);
       
      \end{scope}
      \node [left] at (-0.7, 0) {\scriptsize $DGA \quad H_2:$};
    \end{scope}
  \end{scope}
 
\end{tikzpicture}
\caption{Examples of DFAs and corresponding DGAs, over alphabet
  $\{a, b\}$.}
\label{fig:dfa-dga-examples}
\end{figure}
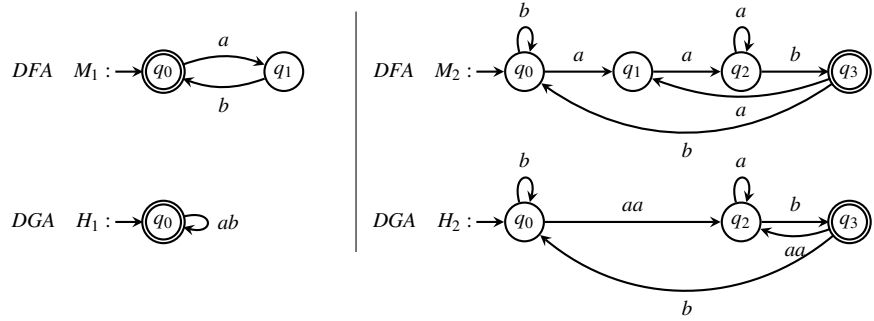

It was
shown in~ \cite{giammarresi1999deterministic} that there is no unique smallest DGA. The paper defines an operation
to suppress states and create longer labels. A state of a DGA is
called \emph{superflous} if it is neither the initial nor final state,
and it has no self-loop. For example, in
Figure~\ref{fig:dfa-dga-examples}, in $M_1$ and $M_2$, state $q_1$ is
superfluous. Such states can be removed, and every pair
$p \xra{\alpha} q$ and $q \xra{\beta} r$ can be replaced with
$p \xra{\alpha \beta} r$.  This operation is extended to a set of
states: given a DGA $H$, a set of states $S$, a DGA $\Ss(H, S)$ is
obtained by suppressing states of $S$, one after the other, in any
arbitrary order. For correctness, there should be no cycle in the
induced subgraph of $H$ restricted to $S$.
The paper proves that minimal DGAs (in number of
states) can be derived by suppressing states, starting from the
canonical DFA.


\section{A new automaton model -- DSA}
\label{sec:new-automaton-model}

We have seen an example of a deterministic suffix automaton in
Figure~\ref{fig:if-else}. A DSA consists of a set of states, and a
finite set of outgoing labels at each state. On an input word $w$, the
DSA finds the earliest prefix which ends with an outgoing label of the
initial state, erases this prefix and goes to the target state of the
transition with the matching label. Now, the DSA processes the rest of
the word from this new state in the same manner. In this section, we
will formally describe the syntax and semantics of DSA.

We start with some more examples. Figure~\ref{fig:example-aab} shows a
DSA for $L_2 = \Sigma^* aab$, the same language as the automata $M_2$
and $H_2$ of Figure~\ref{fig:dfa-dga-examples}. At $q_0$, DSA $\Aa_2$
waits for the first occurrence of $aab$ and as soon as it sees one, it
transitions to $q_3$. Here, it waits for further occurrences of
$aab$. For instance, on the word $abbaabbbaab$, it starts from $q_0$
and reads until $abbaab$ to move to $q_3$. Then, it reads the
remaining $bbaab$ to loop back to $q_3$ and accepts. On a word
$baabaa$, the automaton moves to $q_3$ on $baab$, and continues
reading $aa$, but having nowhere to move, it makes no transition and
rejects the word. Consider another language
$L_3 = \Sigma^*ab\Sigma^*bb$ on the same alphabet $\Sigma$. A similar
machine (as $\Aa_2$) to accept $L_3$ would look like $\Aa_3$ depicted
in Fig.~\ref{fig:example-ab-bb}. For example, on the word $abbbb$, it
would read until $ab$ and move from $q_0$ to $q_1$, read further until
$bb$ and move to $q_2$, then read $b$ and move back to $q_2$ to
accept. We can formally define such machines as automata that
transition on suffixes, or suffix-reading automata.
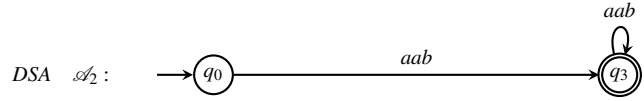
\begin{figure}
  \centering
  \begin{tikzpicture}[state/.style={circle, draw, thick, inner sep =
      2pt}]
    \begin{scope}[every node/.style={state}]
      \node (0) at (0,0) {\scriptsize $q_0$}; \node [double] (1) at
      (5.4,0) {\scriptsize $q_3$};
    \end{scope}
    \begin{scope}[->, thick, >=stealth, auto]
      \draw (-0.75, 0) to (0); \draw (0) to node {\scriptsize $aab$}
      (1); \draw (1) to [loop above] node {\scriptsize $aab$} (1);
      
    \end{scope}
    
    \node [left] at (-1.25, 0) {\scriptsize $DSA \quad \Aa_2:$};
  \end{tikzpicture}
  \caption{DSA $\Aa_2$ accepts $L_2 = \Sigma^*aab$, with
    $\Sigma = \{a, b\}$.}
  \label{fig:example-aab}
\end{figure}

\begin{definition}[DSA]\label{def:suff-reading-aut}
  A \emph{deterministic suffix-reading automaton (DSA)} $\Aa$ is a
  tuple $(Q, \Sigma, q^{init}, \Delta, F)$ where $Q$ is a finite set
  of states, $\Sigma$ is a finite alphabet, $q^{init} \in Q$ is the
  initial state, $\Delta \incl Q \times \Sigma^+ \times Q$ is a finite
  set of transitions, $F \incl Q$ is a set of accepting states.  For a
  state $q \in Q$, we define
  $\out(q) := \{ \a \mid (q, \a, q') \in \Delta \text{ for some } q'
  \in Q \}$ for the set of labels present in transitions out of
  $q$. No state has two outgoing transitions with the same label:
  if $(q, \alpha, q') \in \Delta$ and $(q, \alpha, q'') \in \Delta$,
  then $q' = q''$.
  
  The (total) size $|\Aa|$ of DSA $\Aa$ is defined as the sum of
  the number of states, the number of transitions, and the size
  $|\out(q)|$ for each $q \in Q$, where
  $|\out(q)| := \sum_{\a \in \out(q)} |\a|$.

\end{definition}

\begin{figure}[t]
  \centering
  \begin{tikzpicture}[state/.style={circle, draw, thick, inner sep =
      2pt}]
    \begin{scope}[every node/.style={state}]
      \node (0) at (0,0) {\scriptsize $q_0$}; \node (1) at (2,0)
      {\scriptsize $q_1$}; \node [double] (2) at (4,0) {\scriptsize
        $q_2$};
    \end{scope}
    \begin{scope}[->, thick, >=stealth, auto]
      \draw (-0.75, 0) to (0); \draw (0) to node {\scriptsize $ab$ }
      (1); \draw (1) to [bend left=30] node {\scriptsize $bb$} (2);
      \draw (2) to [loop right] node {\scriptsize $b$} (2); \draw (2)
      to [bend left=30] node {\scriptsize $a$} (1);
    \end{scope}
    \node at (-1.25, 0) {\scriptsize $\Aa_3:$};

    \begin{scope}[xshift=8cm]
      \begin{scope}[every node/.style={state}]
        \node (0) at (0,0) {\scriptsize $q_0$}; \node [double] (1) at
        (2,0) {\scriptsize $q_1$};
      \end{scope}
      \begin{scope}[->, thick, >=stealth, auto]
        \draw (-0.75, 0) to (0); \draw (0) to node {\scriptsize $ab$}
        (1); \draw (0) to [loop above] node {\scriptsize $ba$} (0);
      
      \end{scope}

      \node at (-1.25, 0) {\scriptsize $\Aa_4:$};
    \end{scope}
  \end{tikzpicture}
  \caption{$\Aa_3$ accepts $L_3=\Sigma^* ab \Sigma^*bb$ and $\Aa_4$
    accepts $L_4=(b^*ba)^*a^*ab$.}
  \label{fig:example-ab-bb}
\end{figure}
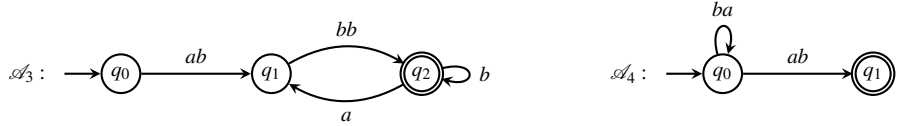

As mentioned earlier, at a state $q$ the automaton waits for a word
that ends with one of its outgoing labels. If more than one label matches,`
then the transition with the longest label is taken.  For example, consider the
DSA in Figure~\ref{fig:if-else}. At
state $s_1$ on reading $fghendif$, both the \texttt{if} and
\texttt{endif} transitions match. The longest match is \texttt{endif}
and therefore the DSA moves to $s_0$. This gives a deterministic
behaviour to the DSA. More precisely: at a state $q$, it reads $w$ to
fire $ (q, \a, q') $ if $\a$ is the longest word in $\out(q)$ which is
a suffix of $w$, and no proper prefix of $w$ has any label in
$\out(q)$ as suffix.  We call this a `move' of the DSA. For example,
consider $\Aa_4$ of Figure~\ref{fig:example-ab-bb} as a DSA. Let us
denote $t := (q_0, ab, q_1)$ and $t':= (q_0, ba, q_1)$. We have moves
$(t, ab)$, $(t, aab)$, $(t, aaab)$, and $(t', ba)$, $(t', bba)$,
etc. In order to make a move on $t$, the word should end with $ab$ and
should have neither $ab$ nor $ba$ in any of its proper prefixes.

\begin{definition}\label{def:DSA-moves}
  A \emph{move} of DSA $\Aa$ is a pair $(t, w)$ where
  $t = (q, \a, q') \in \Delta$ is a transition of $\Aa$ and
  $w \in \Sigma^+$ such that
  \begin{itemize}
  \item 
    $\a$ is the longest word in $\out(q)$ which is a suffix of $w$,
    and
  \item 
    no proper prefix of $w$ contains a label in $\out(q)$ as
    suffix. 
  \end{itemize}
  A move $(t, w)$ denotes that at state $q$, transition $t$ gets
  triggered on reading word $w$. We will also write
  $q \xra[\a]{~w~} q'$ for the move $(t, w)$.
\end{definition}

Whether a word is accepted or rejected is determined by a `run' of the
DSA on it. Naturally the set of words with accepting runs gives the
language of the DSA. Moreover, due to our ``move'' semantics, there is
a unique run for every word.

\begin{definition}
  A run of $\Aa$ on word $w$, starting from a state $q$, is a sequence
  of moves that consume the word $w$, until a (possibly empty) suffix
  of $w$ remains for which there is no move possible: formally, a run
  is a sequence
  $q = q_0 \xra[\a_0]{~w_0~} q_1 \xra[\a_1]{~w_1~} \cdots
  \xra[\a_{m-1}]{~w_{m-1}~} q_m \xra{w_m}$ such that
  $w=w_0 w_1 \dots w_{m-1} w_m$, and $q_m \xra{w_m}$ denotes that
  there is no move using any outgoing transition from $q_m$ on $w_m$
  or any of its prefixes. The run is accepting if $q_m \in F$ and
  $w_m = \e$ (no dangling letters in the end). The language $\Ll(\Aa)$
  of $\Aa$ is the set of all words that have an accepting run starting
  from the initial state $q^{init}$.
\end{definition}

\section{Comparison with DFA and DGA}
\label{sec:comparison-with-dfa}

Every complete DFA can be seen as an equivalent DSA --- since
$\out(q) = \Sigma$ for every state, the equivalent DSA is forced to
move on each letter, behaving like the DFA that we started off
with. For the DSA-to-DFA direction, we associate a specific DFA to
every DSA, as follows. The idea is to replace transitions of a DSA
with a string matching DFA for $\out(q)$ at each
state. Figure~\ref{fig:dsa-to-dfa-eg} gives an example. The
intermediate states correspond to proper prefixes of words in
$\out(q)$.

\begin{definition}[Tracking DFA for a DSA.]\label{def:tracking-dfa}
  For a DSA $\Aa = (Q^\Aa, \Sigma, q_{in}^\Aa, \Delta^\Aa, F^\Aa)$, we
  give a DFA $M_{\Aa}$, called its \emph{tracking DFA}. For
  $q \in Q^\Aa$, let $\outp(q)$ be the set of all prefixes of words in
  $\out(q)$.  States of $M_{\Aa}$ are given by:
  $Q^M = \bigcup_{q \in Q^\Aa}\{ (q, \b) \mid \b \in \outp(q)\} \cup
  q_{copy}$.

  The initial state is $(q^\Aa_{in}, \epsilon)$ and final states are
  $\{ (q, \epsilon) \mid q \in F^\Aa\}$. Transitions are as below: For
  every $q \in Q^\Aa, \b \in \outp(q), a \in \Sigma$, let $\b'$ be the
  longest word in $\outp(q)$ s.t $\b'$ is a suffix of $\b
  a$. 
  \begin{itemize}
  \item $(q, \b) \xra{a} (q', \epsilon)$ if
    $(q,\b', q') \in \Delta^\Aa$, ($q'$ may equal $q$ also)
  \item $(q, \b) \xra{a} (q, \b')$ if $\b' \notin \out(q)$ and
    $\b' \neq \epsilon$,
  \item $(q, \b) \xra{a} q_{copy}$ if $\b' = \epsilon$,
  \item $q_{copy} \xra{a} s$, if $(q, \epsilon) \xra{a} s$ according
    to the above (same outgoing transitions).
  \end{itemize}

\end{definition}

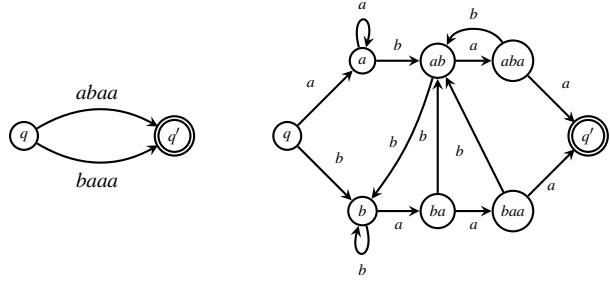
\begin{figure}
  \centering
  \begin{tikzpicture}[state/.style={circle, draw, thick, inner sep =
      2pt}]
    \begin{scope}[every node/.style={state}]
      \node (0) at (0,0) {\tiny $q$}; \node [double] (1) at (2,0)
      {\tiny $q'$};
    \end{scope}
    \begin{scope}[->, >=stealth, thick, auto]
      \draw (0) to [bend left=30] node {\scriptsize $abaa$} (1); \draw
      (0) to [bend right=30] node [below] {\scriptsize $baaa$} (1);
    \end{scope}

    \begin{scope}[xshift=3.5cm]
      \begin{scope}[every node/.style={state}]
        \node (0) at (0,0) {\tiny $q$}; \node (a) at (1,1) {\tiny
          $a$}; \node (b) at (1,-1) {\tiny $b$}; \node (ab) at (2, 1)
        {\tiny $ab$}; \node (ba) at (2,-1) {\tiny $ba$}; \node (aba)
        at (3, 1) {\tiny $aba$}; \node (baa) at (3,-1) {\tiny $baa$};
        \node [double] (1) at (4, 0) {\tiny $q'$};
      
      \end{scope}
      \begin{scope}[->,>=stealth, thick, auto]
        \draw (0) to node {\tiny $a$} (a); \draw (0) to node {\tiny
          $b$} (b); \draw (a) to node {\tiny $b$} (ab); \draw (b) to
        node [below] {\tiny $a$} (ba); \draw (ab) to node {\tiny $a$}
        (aba); \draw (ba) to node [below] {\tiny $a$} (baa); \draw
        (aba) to node {\tiny $a$} (1); \draw (baa) to node [below]
        {\tiny $a$} (1); \draw (b) to [loop below] node {\tiny $b$}
        (0); \draw (a) to [loop above] node {\tiny $a$} (a); \draw
        (ba) to node {\tiny $b$} (ab); \draw (baa) to node {\tiny $b$}
        (ab); \draw (ab) to [bend left=10] node [left] {\tiny $b$}
        (b); \draw (aba) to [bend right=60] node [above] {\tiny $b$}
        (ab);
      \end{scope}

    \end{scope}

  \end{tikzpicture}
  \caption{A DSA on the left, and the corresponding DFA for matching
    the strings $abaa$ and $baaa$.}
  \label{fig:dsa-to-dfa-eg}
\end{figure}

Intuitively, the tracking DFA implements the transition semantics of
DSAs. Starting at $(q, \e)$, the tracking DFA moves along states
marked with $q$ as long as no label of $\out(q)$ is seen as a
suffix. For all such words, the tracking DFA maintains the longest
word among $\outp(q)$ seen as a suffix so far. For instance, in
Figure~\ref{fig:dsa-to-dfa-eg}, at $q$ on reading word $aab$, the DFA
on the right is in state $ab$ (which is the equivalent of $(q, ab)$ in
the tracking DFA definition).

\begin{restatable}{lemma}{trackingDFAEquivalent}
  \label{lem:tracking-dfa-language-equivalent}
  For every DSA $\Aa$, the language $\Ll(\Aa)$ equals the language
  $\Ll(M_\Aa)$ of its tracking DFA.
\end{restatable}

Lemma~\ref{lem:tracking-dfa-language-equivalent} and the fact that
every complete DFA is also a DSA, prove that DSAs recognize regular
languages. We will now compare succinctness of DSA wrt DFA and DGA. 
We start with a
family of languages for which DSAs are concise.

\begin{restatable}{lemma}{dsaSmall}
  \label{lem:dsa-small}
  Let $\Sigma = \{a_1, a_2, \dots, a_n\}$ for some $n \ge 1$. Consider
  the language $L_n = \Sigma^* a_1 a_2 \dots a_n$. There is a DSA for
  this language with size $4 + 2n$.  Any DFA for $L_n$ has size at
  least $n^2$.
\end{restatable}

We now state the final result of this section, which summarizes the
size comparison between DSAs, DFAs, DGAs. For the comparison to DFAs,
we use the fact that every DSA of size $k$ can be converted to its
tracking DFA, which has atmost $2k$ states. Therefore, size of the
tracking DFA is
bounded by $2k$ (states) $+ 2k \cdot |\Sigma|$ (edges)
$+ 2k \cdot |\Sigma|$ (label length), which comes to
$2k ( 1+ 2|\Sigma|)$.

\begin{restatable}{theorem}{comparison}\label{thm:comparing-dsa-with-dfa-dga}
  For a regular language $L$, let
  $n_F^{cmp}, n_F^{trim}, n_G^{trim}, n_S$ denote the size of the
  minimal complete DFA, minimal trim DFA, minimal trim DGA and minimal
  DSA respectively, where size is counted as the sum of the number of
  states, edges and length of edge labels, in all the automata. We
  have:
  \begin{enumerate}
  \item $\dfrac{n_F^{cmp}}{2 (1 + 2|\Sigma|)} \le n_S \le n_F^{cmp}$
  \item no relation between $n_S$ and $n_F^{trim}, n_G^{trim}$: there
    is a language for which $n_S$ is the smallest, and another
    language for which $n_S$ is the largest of the three.
  \end{enumerate}

\end{restatable}


\section{Suffix-tracking sets -- obtaining DSA from DFA}
\label{sec:suffix-tracking-sets}

For DGAs, a
method to derive smaller DGAs by suppressing states was 
recalled in Section~\ref{sec:preliminaries}. Our goal is to
investigate a similar procedure for DSAs. The DSA model creates
new challenges. Suppressing states may not always lead to
smaller automata (in total
size). Figure~\ref{fig:minim-challenges-suppr-states} illustrates an
example where suppressing states leads to an exponentially larger
automaton, due to the exponentially many
paths created. But, suppressing states may sometimes
indeed be useful: in Figure~\ref{fig:suppressing-states-reduces-size},
the DFA on the left is performing a string matching to deduce the
pattern $ab$. On seeing $ab$, it accepts. Any extension is
rejected. This is succinctly captured by the DSA on the right. Notice
that the DSA is obtained by suppressing states $q_1$ and $q_3$. So,
suppressing states may sometimes be useful and sometimes not. In~\cite{giammarresi1999deterministic}, the focus was on
getting a DGA with minimal number of states, and hence suppressing
states was always useful.

More importantly, when can we suppress states? DGAs cannot ``ignore''
parts of the word. This in particular leads to the requirement that a
state with a self-loop cannot be suppressed. DSAs have a more
sophisticated transition semantics. Therefore, the procedure to
suppress states is not as simple. This is the subject of this
section. We deviate from the DGA setting in two ways:  we will select a subset of good states from which we can
construct a DSA (essentially, this means the rest of the states are
suppressed); secondly, our starting point will be complete DFA, on
which we make the choice of states (in DGAs, one could start
with any DGA and suppress states). Our procedure can be broken down
into two steps: (1) Start from a complete DFA, select a subset of states and
    build an induced DSA by connecting states using
    acyclic paths between them; (2) Remove some useless
    transitions.

\begin{figure}
  \centering
  
    \begin{tikzpicture}[state/.style={circle, draw, thick, inner sep =
      2pt}, scale=0.8]
    \begin{scope}[every node/.style={state}]
      \node (0) at (0,0) {\tiny $q_0$};
      \node (1) at (2,0) {\tiny $q_1$};
      \node (2) at (4,0)  {\tiny $q_2$};
      \node [double] (3) at (6,0) {\tiny $q_3$};
    \end{scope}
    \begin{scope}[->, >=stealth, thick]
      \draw (-0.5,0) to (0);
      \draw (0) to [bend left] node [above] {\scriptsize $a$} (1); 
      \draw (0) to [bend right] node [below] {\scriptsize $b$} (1);
       \draw (1) to [bend left] node [above] {\scriptsize $a$} (2); 
       \draw (1) to [bend right] node [below] {\scriptsize $b$} (2);
        \draw (2) to [bend left] node [above] {\scriptsize $a$} (3); 
      \draw (2) to [bend right] node [below] {\scriptsize $b$} (3);
    \end{scope}

    \begin{scope}[xshift=8cm]
       \begin{scope}[every node/.style={state}]
      \node (0) at (0,0) {\tiny $q_0$};
      \node [double] (3) at (4,0) {\tiny $q_3$};
    \end{scope}

    \begin{scope}[->, >=stealth, thick]
      \draw (-0.5, 0) to (0);
      \draw (0) to (3);
    \end{scope}

    \node at (2, 0.4) {\scriptsize $aaa, aab, aba, abb$};
    \node at (2, -0.4) {\scriptsize $baa, bab, bba, bbb$};
    
    \end{scope}
  \end{tikzpicture}
  \caption{Suppressing states can add exponentially many labels and
    increase total size.}
  \label{fig:minim-challenges-suppr-states}
\end{figure}
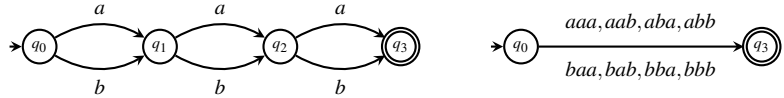

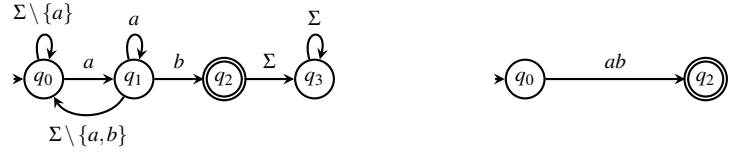
\begin{figure}
  \centering
  \begin{tikzpicture}[state/.style={circle, draw, thick, inner sep =
      2pt}, scale=0.8]
    \begin{scope}[every node/.style={state}]
      \node (0) at (0,0) {\scriptsize $q_0$};
      \node (1) at (1.5,0) {\scriptsize $q_1$};
      \node [double] (2) at (3,0) {\scriptsize $q_2$};
      \node (3) at (4.5,0) {\scriptsize $q_3$};
    \end{scope}
    \begin{scope}[->,>=stealth, thick, auto]
      \draw (-0.5, 0) to (0);
      \draw (0) to [loop above] node {\scriptsize $\Sigma
        \setminus \{a\}$} (0);
      \draw (0) to node {\scriptsize $a$} (1);
      \draw (1) to [bend left = 60] node [below] {\scriptsize $\Sigma
        \setminus \{a, b\}$} (0);
      \draw (1) to [loop above] node {\scriptsize $a$} (1);
      \draw (1) to node {\scriptsize $b$} (2);
      \draw (2) to node {\scriptsize $\Sigma$} (3);
      \draw (3) to [loop above] node {\scriptsize $\Sigma$} (3);
    
    \end{scope}

    \begin{scope}[xshift=8cm]
      \begin{scope}[every node/.style={state}]
        \node (0) at (0,0) {\scriptsize $q_0$};
        \node [double] (2) at (3,0) {\scriptsize $q_2$};
      \end{scope}
      \begin{scope}[->,>=stealth, thick, auto]
        \draw (-0.5, 0) to (0);
        \draw (0) to node {\scriptsize $ab$} (2);
      \end{scope}
    \end{scope}
  \end{tikzpicture}
  \caption{Suppressing states can sometimes reduce total size}
  \label{fig:suppressing-states-reduces-size}
\end{figure}

\paragraph{Building an induced DSA.}

We start with an illustrative example. Consider DFA $M$ in
Figure~\ref{fig:induced-eqv}. 
The DSA on the right of the figure shows such an induced DSA obtained
by marking states $\{q_0, q_2\}$ and connecting them using simple
paths. Notice that the language of the induced DSA and the original
DFA are same in this case. Intuitively, all words that end with an $a$
land in $q_1$. Hence, $q_1$ can be seen to ``track'' the suffix $a$.
Now, consider Figure~\ref{fig:induced-not-eqv}. We do the same trick,
by marking states $\{q_0, q_2\}$ and inducing a DSA. Observe that the
DSA does not accept $aba$, and hence is not language equivalent.  When
does a subset of states induce a language equivalent DSA? Roughly,
this is true when the states that are suppressed track ``suitable
suffixes'' (a reverse engineering of the tracking DFA construction of
Definition~\ref{def:tracking-dfa}). As we will see, the suitable
suffixes will be the simple paths from the selected states to the
suppressed states. We begin by formalizing these ideas and then
present sufficient conditions that ensure language equivalence of the
resulting DSA.

\begin{figure}[t]
  \centering
  \begin{tikzpicture}[state/.style={circle, draw, thick, inner sep =
      2pt},scale=0.8]
    \begin{scope}[every node/.style={state}]
      \node (0) at (0,0) {\scriptsize $q_0$}; \node (1) at (2,0)
      {\scriptsize $q_1$}; \node [double] (2) at (4,0) {\scriptsize
        $q_2$};
    \end{scope}
    \begin{scope}[->, thick, >=stealth, auto]
      \draw (-0.75, 0) to (0); \draw (0) to node {\scriptsize $a$ }
      (1); \draw (1) to node {\scriptsize $b$} (2); \draw (0) to [loop
      above] node {\scriptsize $b$} (0); \draw (1) to [loop above]
      node {\scriptsize $a$} (1); \draw (2) to [loop above] node
      {\scriptsize $a,b$} (2);
    \end{scope}
    \node at (-1.25, 0) {\scriptsize $M:$};

    \begin{scope}[xshift=8cm]
      \begin{scope}[every node/.style={state}]
        \node (0) at (0,0) {\scriptsize $q_0$}; \node [double] (1) at
        (2,0) {\scriptsize $q_2$};
      \end{scope}
      \begin{scope}[->, thick, >=stealth, auto]
        \draw (-0.75, 0) to (0); \draw (0) to node {\scriptsize $ab$}
        (1); \draw (1) to [loop above] node {\scriptsize $a,b$} (1);
        \draw (0) to [loop above] node {\scriptsize $b$} (0);
      
      \end{scope}

      \node at (-1.25, 0) {\scriptsize $\Aa_S:$};
    \end{scope}
  \end{tikzpicture}
  \caption{DFA $M$ and an equivalent DSA $\Aa_S$ `induced' with
    $S = \{q_0, q_2\}$.}
  \label{fig:induced-eqv}
\end{figure}
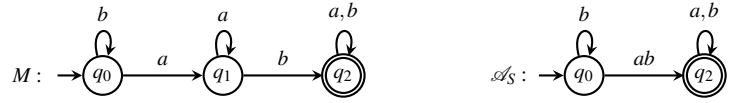

\begin{figure}[t]
  \centering
  \begin{tikzpicture}[state/.style={circle, draw, thick, inner sep =
      2pt}, scale=0.8]
    \begin{scope}[every node/.style={state}]
      \node (0) at (0,0) {\scriptsize $q_0$}; \node (1) at (2,0)
      {\scriptsize $q_1$}; \node [double] (2) at (4,0) {\scriptsize
        $q_2$};
    \end{scope}
    \begin{scope}[->, thick, >=stealth, auto]
      \draw (-0.75, 0) to (0); \draw (0) to node {\scriptsize $a$ }
      (1); \draw (1) to node {\scriptsize $a$} (2); \draw (0) to [loop
      above] node {\scriptsize $b$} (0); \draw (1) to [loop above]
      node {\scriptsize $b$} (1); \draw (2) to [loop above] node
      {\scriptsize $a,b$} (2);
    \end{scope}
    \node at (-1.25, 0) {\scriptsize $M:$};

    \begin{scope}[xshift=8cm]
      \begin{scope}[every node/.style={state}]
        \node (0) at (0,0) {\scriptsize $q_0$}; \node [double] (1) at
        (2,0) {\scriptsize $q_2$};
      \end{scope}
      \begin{scope}[->, thick, >=stealth, auto]
        \draw (-0.75, 0) to (0); \draw (0) to node {\scriptsize $aa$}
        (1); \draw (1) to [loop above] node {\scriptsize $a,b$} (1);
        \draw (0) to [loop above] node {\scriptsize $b$} (0);
      
      \end{scope}

      \node at (-1.25, 0) {\scriptsize $\Aa_S:$};
    \end{scope}
  \end{tikzpicture}
  \caption{DFA $M$ and DSA $\Aa_S$ `induced' with $S = \{q_0,
    q_2\}$. Not equivalent.}
  \label{fig:induced-not-eqv}
\end{figure}
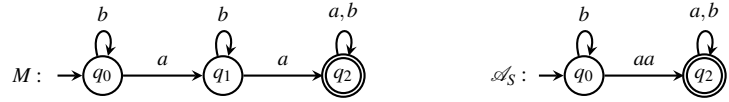

\begin{definition}[Simple words]\label{def:simple-words}
  Consider a complete DFA $M = (Q, \Sigma, q^{init}, \Delta, F)$. Let
  $S \incl Q$ be a subset of states, and $p, q \in Q$. We define
  $\spath{p}{q}{S}$, the \emph{simple words from $p$ to $q$ modulo
    $S$}, as the set of all words $a_1 a_2 \dots a_n \in \Sigma^+$
  such that there is a path:
  $p = p_0 \xra{a_1} p_1 \xra{a_2} \cdots p_{n-1} \xra{a_n} p_n = q$
  in $M$ where
  \begin{itemize}
  \item no intermediate state belongs to $S$:
    $\{ p_1, \dots, p_{n-1}\} \incl Q \setminus S$, and
  \item there is no intermediate cycle: if $p_i = p_j$ for some
    $0 \le i < j \le n$, then $p_i = p_0$ and $p_j = p_n$.
  \end{itemize}
  We write $\spaths{p}{S}$ for $\bigcup_{q \in Q} \spath{p}{q}{S}$,
  the set of all simple words modulo $S$, emanating from $p$.
\end{definition}

For example, in Figure~\ref{fig:induced-eqv}, with $S = \{q_0, q_2\}$,
we have $\spath{q_0}{q_1}{S} = \{a\}$, $\spath{q_0}{q_0}{S} = \{b\}$
and $\spath{q_0}{q_2}{S} = ab$. These are the same in
Figure~\ref{fig:induced-not-eqv}, except $\spath{q_0}{q_2}{S} = aa$.

Fix a complete DFA $M$ for this section. A DSA can be `induced' from
$M$ using $S$, by fixing states to be $S$ (initial and final states
retained) and transitions to be the simple words modulo $S$ connecting
them i.e. $p \xra{\s} q$ if $\s \in \spath{p}{q}{S}$ (Figure
\ref{fig:induced-eqv}).

\begin{definition}[Induced DSA]\label{def:induced-dsa}
  Given a DFA $M$ and a set $S$ of states in $M$ that contains the
  initial and final states, we define the induced DSA of $M$ (using
  $S$). The states of the induced DSA are given by $S$. The initial
  and final states are the same as in $M$. The transitions are given
  by the simple words modulo $S$ i.e. $p \xra{\s} q$ if
  $\s \in \spath{p}{q}{S}$, for every pair of states $p, q \in S$.
\end{definition}

The induced DSA may not be language-equivalent (Figure
\ref{fig:induced-not-eqv}); to ensure that, we need to check some
conditions. Here is a central definition.

\begin{definition}[Suffix-compatible transitions]\label{def:suffix-compatibility}
  Fix a subset $S \incl Q$. A transition $q \xra{a} u$ is
  suffix-compatible w.r.t. $S$ if either of
  $q,u \in S~\textbf{OR}~\forall p \in S$, and for every
  $\s \in \spath{p}{q}{S}$, there is an $\a \in \spath{p}{u}{S}$ s.t.:
  \begin{itemize}
  \item $\a$ is a suffix of $\s a$, and
  \item moreover, $\a$ is the longest suffix of $\s a$ among words in
    $\spaths{p}{S}$.
  \end{itemize}
\end{definition}

Note that a transition $q \xra{a} u$ is trivially suffix-compatible if
$q \in S$ or $u \in S$. The rest of the condition only needs to be
checked when both of $q,u \notin S$. In
Figure~\ref{fig:induced-not-eqv}, we find the self-loop at $q_1$ to
not be suffix-compatible: we have $S = \{q_0, q_2\}$, and
$\spath{q_0}{q_1}{S} = \{ a \}$, $\spaths{q_0}{S} = \{b, a, ab\}$; the
transition $q_1 \xra{b} q_1$ is not suffix-compatible since there is
no suffix of $ab$ in $\spath{q_0}{q_1}{S}$. Whereas in
Figure~\ref{fig:induced-eqv}, the loop is labeled $a$ instead of
$b$. The transition $q_1 \xra{a} q_1$ is suffix-compatible, since the longest suffix of $aa$ among $\spaths{q_0}{S}$ is $a$ and it is
present in $\spath{q_0}{q_1}{S}$. Let us take the DFA in the right of
Figure~\ref{fig:dsa-to-dfa-eg}, and let $S = \{q, q'\}$. Here are some
of the simple path sets: $\spath{q}{ab}{S} = \{ab, bab, baab\}$,
$\spath{q}{aba}{S} = \{aba, baba, baaba\}$. Consider the transition
$aba \xra{b} ab$. It can be verified that for every
$\s \in \spath{q}{aba}{S}$, the longest suffix of the extension
$\s a$, among simple paths out of $q$, indeed lies in the state $ab$.
In fact, all transitions satisfy suffix-compatibility w.r.t. the
chosen set $S$.

The suffix-compatibility condition is described using simple paths to
states. It requires that every transition take each simple word
reaching its source to the state tracking the longest suffix of its
one-letter extension. This condition on simple paths, transfers to all
words, that circle around the suppressed states. In
Figure~\ref{fig:dsa-to-dfa-eg}, this property can be verified by
considering the word $bbabab$ and its run: $q \xra{b} b \xra{b} b \xra{a} ba \xra{b} ab \xra{a} aba \xra{b} ab$.
At each step, the state reached corresponds to the longest suffix
among the simple words out of $q$.
In the next two lemmas, we prove this claim.

We will use a special
notation: for a state $p \in S$, we write $\out(p, S)$ for
$\bigcup_{r \in S} \spath{p}{r}{S}$; these are the simple words that
start at $p$ and end in some state $r$ of $S$. Notice that these are
the words that appear as transitions in the induced DSA. In
particular, $\out(p)$ in the induced DSA equals $\out(p, S)$.

\begin{restatable}{lemma}{sfxCompatibleDfaToWords}
  \label{lem:sfx-compatibility-dfa-paths-to-words}
  Let $S$ be a set of states such that every transition of $M$ is
  suffix-compatible w.r.t. $S$. Pick $p \in S$, and let
  $w \in \Sigma^+$ be a word with a run
  $p = p_0 \xra{w_1} p_1 \xra{w_2} p_2 \dots p_{n-1} \xra{w_n} p_n$
  such that the intermediate states $p_1, \dots, p_{n-1}$ belong to
  $Q \setminus S$. The state $p_n$ may or may not be in $S$. Then:
  \begin{itemize}
  \item no proper prefix of $w$ contains any word from $\out(p,S)$ as
    suffix, and
  \item there is $\a \in \spath{p}{p_n}{S}$ such that $\a$ is the
    longest suffix of $w$ among words in $\spaths{p}{S}$.
  \end{itemize}
\end{restatable}

\begin{restatable}{lemma}{sfxCompatibleWordsToDfa}
  \label{lem:sfx-compatibility-words-to-paths}
  Let $S$ be a set of states such that every transition of $M$ is
  suffix-compatible w.r.t. $S$. Let $p \in S$, and $w \in \Sigma^+$ be
  a word such that no proper prefix of $w$ contains a word in
  $\out(p,S)$ as suffix. Then:
  \begin{itemize}
  \item The run of $M$ starting from $p$, is of the form
    $p \xra{w_1} p_1 \xra{w_2} p_2 \dots p_{n-1} \xra{w_n} p_n$ where
    $\{p_1, \dots, p_{n-1}\} \incl Q \setminus S$ (notice that we have
    not included $p_n$, which may or may not be in $S$).
  \item the longest suffix of $w$, among $\spaths{p}{S}$ lies in
    $\spath{p}{p_n}{S}$.
  \end{itemize}

\end{restatable}

Suffix-compatibility alone does not suffice to preserve the
language. In Figure~\ref{fig:well-formed-set}, consider
$S = \{0, 2, 4\}$. Every transition is suffix-compatible
w.r.t. $S$. The DSA induced using $S$ is shown in the middle. Notice
that it is not language equivalent, due to the word $aba$ for
instance. The run of $aba$ looks as follows: $0 \xra{ab} 4 \xra{b}
4$. The expected run was $0 \xra{aba} 2$, but that does not happen
since there is a shorter prefix with a matching transition. Even
though, we have suffix-compatibility, we need to ensure that there are
no ``conflicts'' between outgoing patterns. This leads to the next
definition.

\begin{definition}[Well-formed set]\label{def:well-formed-set}
  A set of states $S \incl Q$ is well-formed if there is no
  $p \in S, q \in S$ and $q' \notin S$, with a pair of words
  $\a \in \spath{p}{q}{S}$ (simple word to a state in $S$) and
  $\beta \in \spath{p}{q'}{S}$ (simple word to a state not in $S$)
  such that $\alpha$ is a suffix of $\beta$.
\end{definition}

We observe that the set $S=\{0,2,4\}$ is not well-formed since
$b \in \spath{0}{4}{S}, ab \in \spath{0}{3}{S}$ and $b$ is a suffix of
$ab$. Whereas $S'=\{0,2,3,4\}$ is both suffix-tracking, and
well-formed, and induces an equivalent DSA. On the word $aba$, the run
on the DSA would be $0 \xra{ab} 3 \xra{a} 2$. The first move
$0 \xra{ab} 3$ applies the longest match criterion, and the transition
since $ab$ is a longer suffix than $b$.  This was not possible before
since $3 \notin S$. It turns out that the two conditions --- suffix-compatibility and
well-formedness --- are sufficient to induce a language equivalent
DSA.

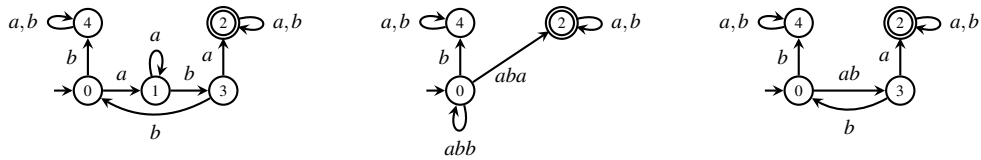
\begin{figure}[t]
  \centering
  \begin{tikzpicture}[state/.style={circle, draw, thick, inner sep =
      2pt}, scale=0.9]
    \begin{scope}[every node/.style={state}]
      \node (0) at (-0.5,0) {\tiny $\rr$}; \node (1) at (0.5, 0)
      {\tiny $\bb$}; \node (2) at (1.5, 0) {\tiny $\yy$}; \node
      [double] (3) at (1.5, 1) {\tiny $\gg$}; \node (4) at (-0.5, 1)
      {\tiny $\ww$};
    \end{scope}
    \begin{scope}[->, thick, >=stealth, auto]
      \draw (-1, 0) to (0); \draw (0) to node {\scriptsize $a$} (1);
      \draw (0) to node {\scriptsize $b$} (4); \draw (1) to [loop
      above] node {\scriptsize $a$} (1); \draw (1) to node
      {\scriptsize $b$} (2); \draw (2) to node {\scriptsize $a$} (3);
      \draw (4) to [loop left] node {\scriptsize $a, b$} (4); \draw
      (3) to [loop right] node {\scriptsize $a, b$} (3); \draw (2) to
      [bend left=30] node {\scriptsize $b$} (0);
    \end{scope}

    \begin{scope}[xshift=5 cm]
      \begin{scope}[every node/.style={state}]
        \node (0) at (0,0) {\tiny $\rr$}; \node [ double] (3) at (1.5,
        1) {\tiny $\gg$}; \node (4) at (0, 1) {\tiny $\ww$};
      \end{scope}
      \begin{scope}[->, thick, >=stealth, auto]
        \draw (-0.5, 0) to (0); \draw (0) to node [below] {\scriptsize
          $aba$} (3); \draw (0) to node {\scriptsize $b$} (4); \draw
        (0) to [loop below] node {\scriptsize $abb$} (0); \draw (4) to
        [loop left] node {\scriptsize $a, b$} (4); \draw (3) to [loop
        right] node {\scriptsize $a, b$} (3);
      
      \end{scope}

    \end{scope}

    \begin{scope}[xshift=10cm]
      \begin{scope}[every node/.style={state}]
        \node (0) at (0,0) {\tiny $\rr$}; \node [double] (3) at (1.5,
        1) {\tiny $\gg$}; \node (2) at (1.5, 0) {\tiny $\yy$}; \node
        (4) at (0, 1) {\tiny $\ww$};
      \end{scope}
      \begin{scope}[->, thick, >=stealth, auto]
        \draw (-0.5, 0) to (0); \draw (0) to node {\scriptsize $ab$}
        (2); \draw (0) to node {\scriptsize $b$} (4); \draw (2) to
        node {\scriptsize $a$} (3); \draw (4) to [loop left] node
        {\scriptsize $a, b$} (4); \draw (3) to [loop right] node
        {\scriptsize $a, b$} (3); \draw (2) to [bend left=30] node
        {\scriptsize $b$} (0);
      \end{scope}

    \end{scope}
  \end{tikzpicture}
  \caption{A DFA, a non-equivalent DSA and an equivalent induced DSA. }
  \label{fig:well-formed-set}
\end{figure}

\begin{definition}[Suffix-tracking sets]\label{def:suffix-tracking-set}
  A set of states $S \incl Q$ is suffix-tracking if it contains the
  initial and accepting states, and
  \begin{enumerate}
  \item every transition of $M$ is suffix-compatible w.r.t. $S$,
  \item and $S$ is well-formed.
  \end{enumerate}
\end{definition}

All these notions lead to the main theorem of this section.

\begin{restatable}{theorem}{suffixTrackingCorrect}
  \label{thm:suffix-tracking-is-correct}
  Let $S$ be a suffix-tracking set of complete DFA $M$, and let
  $\Aa_S$ be the DSA induced using $S$. Then: $\Ll(\Aa_S) = \Ll(M)$
\end{restatable}
\begin{proof}
  Pick $w \in \Ll(M)$. There is an accepting run
  $q_0 \xra{w_1} q_1 \xra{w_2} \dots \xra{w_n} q_n$ of $M$ on $w$. By
  Definition~\ref{def:suffix-tracking-set}, we have $q_0, q_n \in
  S$. Let $1 \le i \le n$ be the smallest index greater than $0$, such
  that $q_i \in S$. Consider the run segment
  $q_0 \xra{w_1} q_1 \xra{w_2} \dots \xra{w_i} q_i$. By
  Lemma~\ref{lem:sfx-compatibility-dfa-paths-to-words}, and by the
  definition of induced DSA~\ref{def:induced-dsa}, no transition of
  $\Aa_S$ out of $q_0$ is triggered until $w_1 \dots w_{i-1}$, and
  then on reading $w_i$, the transition $q_0 \xra{\a} q_i$ is
  triggered, where $\a \in \spath{p}{q}{S}$, and $\a$ is also the
  longest suffix of $w_1 \dots w_i$ among $\spaths{p}{S}$. In
  particular, it is the longest suffix among outgoing labels from
  $q_0$ in $\Aa_S$. This shows there is a move
  $q_0 \xra[\a]{w_1 \dots w_i} q_i$ in $\Aa_S$. Repeat this argument
  on rest of the run
  $q_i \xra{w_{i+1}} q_{i+1} \xra{w_{i+1}} \dots \xra{w_n} q_n$ to
  extend the run of $\Aa_S$ on the rest of the word. This shows
  $w \in \Ll(\Aa_S)$.
  
  Pick $w \in \Ll(\Aa_S)$. There is an accepting run $\rho$ of $\Aa_S$
  starting at the initial state $q_0$. Consider the first move
  $q_0 \xra[\a]{w_1 \dots w_i} q_i$ of $\Aa_S$ on the word. By the
  semantics of a move (Definition~\ref{def:DSA-moves}) and
  Lemma~\ref{lem:sfx-compatibility-words-to-paths}, we obtain a run
  $q_0 \xra{w_1} q_1 \xra{w_2} \dots q_{i-1} \xra{w_i} q_i$ of $M$
  where the intermediate states $q_1, \dots, q_{i-1}$ lie in
  $Q \setminus S$. We apply this argument for each move $\rho$ in the
  accepting run of $\Aa_S$ to get an accepting run of $M$.
\end{proof}

\paragraph{Removing some useless transitions.}
Let us now get back to Figure~\ref{fig:dsa-to-dfa-eg} to see if we can
derive the DSA on the left from the DFA on the right (assuming $q$ is
the initial state). As seen earlier, the set $S = \{q, q'\}$ is
suffix tracking. It is
also well formed since $baaa$ is not a suffix of any prefix of $abaa$
and vice-versa. The DSA $\Aa_S$ induced using $q$ and $q'$ will have
the set of words in $\spath{q}{q'}{S}$ as transitions between $q$ and
$q'$. Both $abaa$ and $baaa$ belong to $\spath{q}{q'}{S}$. However,
there are some additional simple words: for instance, $abbaaa$. Notice
that $baaa$ is a suffix of $abbaaa$, and therefore even if we remove the
transition on $abbaaa$, there will be a move to $q'$ via
$q \xra{baaa} q'$. This tempts us to use only the suffix-minimal words
in the transitions of the induced DSA. This is not always safe, as we
explain below. We show how to carefully remove
``bigger-suffix-transitions''.

Consider the DSA on the left in
Figure~\ref{fig:bigger-suffix-transitions}. If $caba$ is removed, the
moves which were using $caba$ can now be replaced by $ba$ and we still
have the same pair of source and target states. Consider the picture
on the right of the same figure. There is an outgoing edge to a
different state on $aba$. Suppose we remove $caba$. The word $caba$
would then be matched by the longer suffix $aba$ and move to a
different state.
Another kind of useless transitions are some of the self-loops on
DSAs. In Figure~\ref{fig:induced-eqv}, the self-loop on $b$ at $q_0$
can be removed, without changing the language. This can be generalized
to loops over longer words, under some conditions.

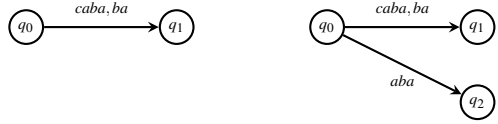
\begin{figure}[t]
  \centering
  \begin{tikzpicture}[state/.style={circle, draw, thick, inner sep =
      2pt}]
    \begin{scope}[every node/.style={state}]
      \node (0) at (0,0) {\tiny $q_0$}; \node (1) at (2,0) {\tiny
        $q_1$};
    \end{scope}
    \begin{scope}[->, >=stealth, thick, auto]
      \draw (0) to node {\tiny $caba, ba$} (1);
    \end{scope}

    \begin{scope}[xshift=4cm]
      \begin{scope}[every node/.style={state}]
        \node (0) at (0,0) {\tiny $q_0$}; \node (1) at (2,0) {\tiny
          $q_1$}; \node (2) at (2,-1) {\tiny $q_2$};
      \end{scope}
      \begin{scope}[->, >=stealth, thick, auto]
        \draw (0) to node {\tiny $caba, ba$} (1); \draw (0) to node
        [below] {\tiny $aba$} (2);
      \end{scope}
    \end{scope}
  \end{tikzpicture}
  \caption{Illustrating bigger-suffix transitions and when they are
    useless}
  \label{fig:bigger-suffix-transitions}
\end{figure}

\begin{definition}\label{def:useless-transitions}
  Let $\Aa$ be a DSA, $q, q'$ be states of $\Aa$ and
  $t:= q \xra{\a} q'$ be a transition.

  We call $t$ a \emph{bigger-suffix-transition} if there exists
  another transition $(q, \b, q')$ with $\beta$ a suffix of $\alpha$. 

  If there is a transition $t' := q \xra{\gamma} q'' ~ (q'' \neq q')$,
  such that $\beta$ is a suffix of $\gamma$, and $\gamma$ is a suffix
  of $\alpha$, we call $t$ \emph{useful}. A
  bigger-suffix-transition is called \emph{useless} if it is not
  useful.

  We will say that $t$ is a \emph{useless self-loop} if $q = q'$,
  $q$ is not an accepting state, and no suffix of $\a$ is a prefix of
  some outgoing label in $\out(q)$.
\end{definition}

In Figure~\ref{fig:bigger-suffix-transitions}, for the automaton on
the left, the transition on $caba$ is useless. Whereas for the DSA on
the right, $caba$ is a bigger-suffix-transition, but it is
useful. The
self-loop on $q_0$ in Figure~\ref{fig:induced-eqv} is
useless, but the loop on $q_0$ in Figure~\ref{fig:example-ab-bb}
is useful. Lemmas~\ref{lem:bigger-suffix-transitions-useless}
and~\ref{lem:removable-self-loop-useless} 
prove correctness of
removing useless transitions.

\begin{lemma}
  \label{lem:bigger-suffix-transitions-useless}
  Let $\Aa$ be a DSA, and let $t:= q \xra{\a} q'$ be a useless
  bigger-suffix-transition. Let $\Aa'$ be the DSA obtained by removing
  $t$ from $\Aa$. Then, $L(\Aa) = L(\Aa')$.
\end{lemma}
\begin{proof} 
  
 \emph{To show $L(\Aa) \incl L(\Aa')$.} Let $w \in L(\Aa)$ and let
  $q_0 \xra[\a_0]{w_0} q_1 \xra[\a_1]{w_1} \cdots
  \xra[\a_{m-1}]{w_{m-1}} q_m$ be an accepting run. If no
  $(q_i, \a_i, q_{i+1})$ equals $(q, \a, q')$, then the same run is
  present in $S'$, and hence $w \in L(S')$. Suppose
  $(q_j, \a_j, q_{j+1}) = (q, \a, q')$ for some $j$. So, the word
  $w_{j}$ ends with $\a$. As $(q, \a, q')$ is a
  bigger-suffix-transition, there is another $(q, \b, q')$ such that
  $\b \sfx \a$. Therefore, the word $w_j$ also ends with $\b$. Since
  there was no transition matching a proper prefix of $w_j$, the same
  will be true at $\Aa'$ as well, since it has fewer transitions. It
  remains to show that $q_j \xra[\b]{w_j} q_{j+1}$ is a move. The only
  way this cannot happen is if there is a $q \xra{\gamma} q''$ with
  $\b \sfx \gamma \sfx \a$. But this is not possible since
  $q \xra{\a} q'$ is a useless bigger-suffix transition. Therefore,
  every move using $(q, \a, q')$ in $\Aa$ will now be replaced by
  $(q, \b, q')$ in $\Aa'$. Hence we get an accepting run in $\Aa'$,
  implying $w \in L(\Aa')$.

  \emph{To show $L(\Aa') \incl L(\Aa)$}. Consider $w \in L(\Aa')$ and
  an accepting run
  $q_0 \xra[\a_0]{w_0} q_1 \xra[\a_1]{w_1} \cdots
  \xra[\a_{m-1}]{w_{m-1}} w_m$ in $\Aa'$. Notice that if
  $q \xra[\b]{w_j} q'$ is a move in $\Aa'$, the same is a move in
  $\Aa$ when $\a \not\sfx w_j$. When $\a \sfx w_j$, then the
  bigger-suffix-transition $q \xra{\a} q'$ will match and the move $q
  \xra[\b]{w_j} q'$ 
  gets replaced by $q \xra[\a]{w_j} q'$. Hence we will get the same
  run, except that some of the moves using $q \xra{\b} q'$ may get
  replaced with $q \xra{\a} q'$.
\end{proof}

For the correctness of removing useless self-loops, we assume that the
DFA that we obtain is well-formed
(Definition~\ref{def:well-formed-dsa}) and has no useless
bigger-suffix-transitions. The induced DSA that we obtain from
suffix-tracking sets is indeed well-formed. Starting from this induced
DSA, we can first remove all useless bigger-suffix-transitions, and
then remove the useless self-loops.

\begin{lemma}
  \label{lem:removable-self-loop-useless}
  Let $\Aa$ be a well-formed DSA that has no removable
  bigger-suffix-transitions. Let $t := (q, \a, q)$ be a removable
  self-loop. Then the DSA $\Aa'$ obtained by removing $t$ from $\Aa$
  satisfies $\Ll(\Aa) = \Ll(\Aa')$.
\end{lemma}
\begin{proof}
 \emph{To show $\Ll(\Aa) \incl \Ll(\Aa')$}. Let $w \in \Ll(\Aa)$ and
  let $\rho:= q_0 \xra{w_0} q_1 \xra{w_1} \cdots \xra{w_{m-1}} q_m$ be
  an accepting run. Suppose $t$ matches the segment
  $q_j \xra{w_j} q_{j+1}$. Hence $q_j = q_{j+1} = q$.  Observe that as
  $q$ is not accepting, we have $j+1 \neq m$. Therefore there is a
  segment $q_{j+1} \xra{w_{j+1}} q_{j+2}$ in the run. We claim that if
  $t$ is removed, then no transition out of $q$ can match any prefix
  of $w_j w_{j+1}$.

  First we see that no prefix of $w_j$ can be
  matched, including $w_j$ itself: if at all there is a match, it
  should be at $w_j$, and a $\beta$ that is smaller than $\a$. By
  assumption, $\a$ is not a removable
  bigger-suffix-transition. Therefore, there is a transition $q
  \xra{\gamma} q'$, with $\beta \sfx \gamma \sfx \alpha$. This
  contradicts the assumption that $\alpha$ is a removable
  self-loop. Therefore there is no match upto $w_j$.

  Suppose some $(q, \b, q')$
  matches a prefix $w_j u$ such that $\b = v u$, that is, $\b$
  overlaps both $w_j$ and $w_{j+1}$. If $\a \sfx v$, then it violates
  well-formedness of $S$ since it would be a suffix of a proper prefix
  ($v$) of $\b$. This shows $v \sfx \a$ (since both are suffixes of
  $w_j$) and $v \prfx \b$, contradicting the assumption that $t$ is
  removable. Therefore, $\b$ does not overlap $w_j$. But then, if $\b$
  is a suffix of a proper prefix of $w_{j+1}$, we would not have the
  segment $q_{j+1} \xra{w_{j+1}} q_{j+2}$ in the run
  $\rho$. Therefore, the only possibility is that we have a segment
  $q_j \xra{w_j w_{j+1}} q_{j+2}$. We have fewer occurrences of the
  removable loop $(q, \a, q)$ in the modified run. Repeating this
  argument for every match of $(q, \a, q)$ gives an accepting run of
  $\Aa'$. Hence $w \in L(\Aa')$.

  \emph{To show $L(\Aa') \incl L(\Aa)$.} Let $w \in L(\Aa')$ and
  $\rho' := q_0 \xra{w_0} q_1 \xra{w_1} \cdots \xra{w_{m-1}} q_m$ be
  an accepting run in $\Aa'$. Suppose $q_j \xra{w_j} q_{j+1}$ is
  matched by $(q, \b, q')$. Let $w_j = v u $ with $\a \sfx v$. Then
  the removable-self-loop $(q, \a, q)$ will match the prefix
  $v$. Suppose $\beta$ overlaps with both $v$ and $u$, that is $\beta
  = \beta' u$. We cannot have $\alpha \sfx \beta'$ due to
  well-formedness of $\Aa$. We cannot have $\beta' \sfx \alpha$ since
  this would mean there is a suffix of $\alpha$ which is a prefix of $\beta$,
  violating the removable-self-loop condition.  Therefore,
  $\b$ is entirely inside $u$, that is, $\b \sfx
  u$. Hence in $\Aa$ the run will first start with $q \xra{v}
  q$. Applying the same
  argument, prefixes of the remaining word where 
  $t$ matches will be matched until there is a part of the word where
  $(q, \b, q')$ matches. This applies to every segment, thereby giving
  us a run in $\Aa$.
\end{proof}

We now get to the core definition of this section, which tells how to
derive a DSA from a DFA, using the methods developed so far.

\begin{definition}[DFA-to-DSA derivation] \label{def:derived-dsa}
  A DSA is said to be \emph{derived from} DFA $M$ using
  $S \subseteq Q$, if it is identical to an induced DSA of $M$
  (using $S$) with all useless transitions removed.
\end{definition}

By Theorem~\ref{thm:suffix-tracking-is-correct} and Lemma
\ref{lem:bigger-suffix-transitions-useless}, we get the following
result.

\begin{theorem}
  Every DSA that is derived from a complete DFA is language equivalent
  to it.
\end{theorem}


\section{Minimality, some observations and some challenges}
\label{sec:minim-some-observ}

Theorem~\ref{thm:comparing-dsa-with-dfa-dga} shows that we can not
expect DSAs to be smaller than (trim) DFAs or DGAs in
general. However, Lemma~\ref{lem:dsa-small} and
Figure~\ref{fig:if-else} show that there are cases where DSAs are
smaller and more readable. This motivates us to ask the question of
how we can find a minimal DSA, that is, a DSA of the smallest (total)
size. The first observation is that minimal DSAs need not be unique
--- see Figure~\ref{fig:no-canonical}. The next simple observation
is that a minimal DSA will not have useless transitions
since removing them gives an equivalent DSA with strictly smaller
size.  In fact, we can assume a certain well-formedness condition on
the minimal DSAs, in the same spirit as the definition of well-formed
sets in our derivation procedure: if there are two transitions
$q \xra{\a} q_1$ and $q \xra{\b_1 \a \b_2} q_2$, then we can remove
the second transition since it will never get fired.

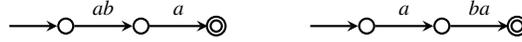
\begin{figure}[t]
  \centering
  \begin{tikzpicture}[state/.style={circle, draw, thick, inner sep =
      2pt}]
    \begin{scope}[every node/.style={state}]
      \node (0) at (0,0) {}; \node (1) at (1, 0) {}; \node [double]
      (2) at (2, 0) {};
    \end{scope}
    \begin{scope}[->,>=stealth, thick, auto]
      \draw (-0.75, 0) to (0); \draw (0) to node {\scriptsize $ab$}
      (1); \draw (1) to node {\scriptsize $a$} (2);
    \end{scope}
    \begin{scope}[xshift=4cm]
      \begin{scope}[every node/.style={state}]
        \node (0) at (0,0) {}; \node (1) at (1, 0) {}; \node [double]
        (2) at (2, 0) {};
      \end{scope}
      \begin{scope}[->,>=stealth, thick, auto]
        \draw (-0.75, 0) to (0); \draw (0) to node {\scriptsize $a$}
        (1); \draw (1) to node {\scriptsize $ba$} (2);
      \end{scope}
    \end{scope}
  \end{tikzpicture}
  \caption{Minimal DSA is not unique}
  \label{fig:no-canonical}
\end{figure}

\begin{definition}[Well-formed DSA]\label{def:well-formed-dsa}
  A DSA $\Aa$ is \emph{well-formed} if for every state $q$, no
  outgoing label $\a \in \out(q)$ is a suffix of some proper prefix
  $\b'$ of another outgoing label $\beta \in \out(q)$.
\end{definition}

Any transition violating well-formedness can be removed, without
changing the language. Therefore, we can safely assume that minimal
DSAs are well-formed. Due to the ``well-formedness'' property in
suffix-tracking sets, the DSAs induced by suffix-tracking sets are
naturally well-formed. Since removing useless transitions preserves
this property, the DSAs that are derived using our DFA-to-DSA
procedure (Definition~\ref{def:derived-dsa}) are well-formed. The next
proposition shows that every DSA that is well-formed and has no
useless transitions (and in particular, the minimal
DSAs) can be derived from the corresponding tracking DFAs.

\begin{restatable}{proposition}{dsaDerivableFromTracking}
  \label{prop:dsa-derivable-from-tracking-dfa}
  Every well-formed DSA with no useless transitions can
  be derived from its tracking DFA.
\end{restatable}

Proposition~\ref{prop:dsa-derivable-from-tracking-dfa} says that if we
somehow had access to the tracking DFA of a minimal DSA, we will be
able to derive it using our procedure. The challenge however is that
this tracking DFA may not necessarily be the canonical DFA for the
language. In fact, we now show that a smallest DSA that can be derived
from the canonical DFA need not be a minimal DSA.

\begin{figure}
  \begin{tikzpicture}[state/.style={circle, draw, thick, inner sep =
      2pt}]
    \begin{scope}[every node/.style={state}]
      \node (0) at (0,0) {\tiny $q_0$}; \node (1) at (2,1) {\tiny
        $q_1$}; \node (2) at (4,1) {\tiny $q_2$}; \node [double] (3)
      at (4, 0) {\tiny $q_4$}; \node (45) at (3,-1) {\tiny $p$};
    \end{scope}

    \begin{scope}[->,>=stealth, thick, auto]
      \draw (-0.75, 0) to (0); \draw (0) to [loop above] node [left]
      {\tiny $\Sigma \setminus \{a, b\}$} (0); \draw (0) to node
      [below] {\tiny $a$} (1); \draw (1) to [bend right=30] node
      [left] {\tiny $\Sigma \setminus \{a, b\}$} (0); \draw (1) to
      [loop above] node {\tiny $a$} (1); \draw (1) to node {\tiny $b$}
      (2); \draw (2) to [bend left=20] node [below, near start] {\tiny
        $\Sigma \setminus a$} (0); \draw (2) to node {\tiny $a$} (3);
      \draw (0) to node {\tiny $b$} (45);
      
      \draw (45) to node {\tiny $b$} (3);
 
      \draw (45) to 
      [bend left] node {\tiny $\Sigma \setminus \{a,b\}$} (0); \draw
      (45) to [loop below] node {\tiny $a$} (45); \draw (3) to [loop
      right] node {\tiny $\Sigma$} (3);
    \end{scope}

    \begin{scope}[xshift=8cm]
      \begin{scope}[every node/.style={state}]
        \node (0) at (0,0) {\tiny $q_0$}; \node (2) at (4,1) {\tiny
          $q_2$}; \node [double] (3) at (4, 0) {\tiny $q_4$}; \node
        (45) at (3,-1) {\tiny $p$};
      \end{scope}
      
      \begin{scope}[->, >=stealth, thick, auto]
        \draw (-0.75, 0) to (0);
        
        \draw (0) to [bend left] node {\tiny $ab$} (2);
        
        \draw (2) to [bend left=20] node [below, near start] {\tiny
          $\Sigma \setminus a$} (0); \draw (2) to node {\tiny $a$}
        (3); \draw (0) to node {\tiny $b$} (45);
        
        \draw (45) to node {\tiny $b$} (3);
 
        \draw (45) to 
        [bend left] node {\tiny $\Sigma \setminus \{a, b\}$}
        (0); 
        \draw (3) to [loop right] node {\tiny $\Sigma$} (3);
      \end{scope}
    \end{scope}
  \end{tikzpicture}
  \caption{DFA $M^*$ on the left and a derived DSA $\Aa^*_S$ with
    $S = \{q_0, q_2, q_4, p\}$ on the right.}
  \label{fig:dfa-m-star-dsa}
\end{figure}
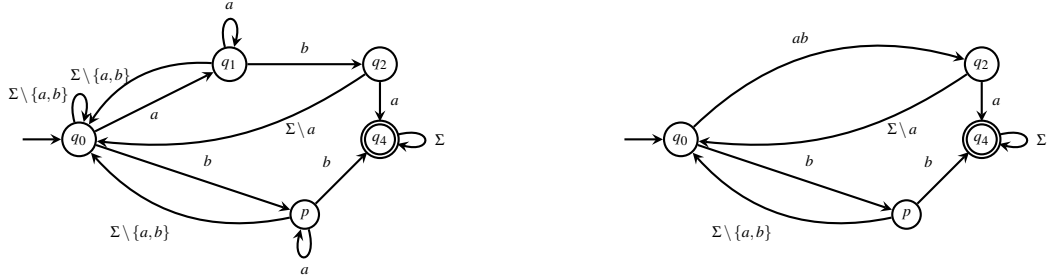

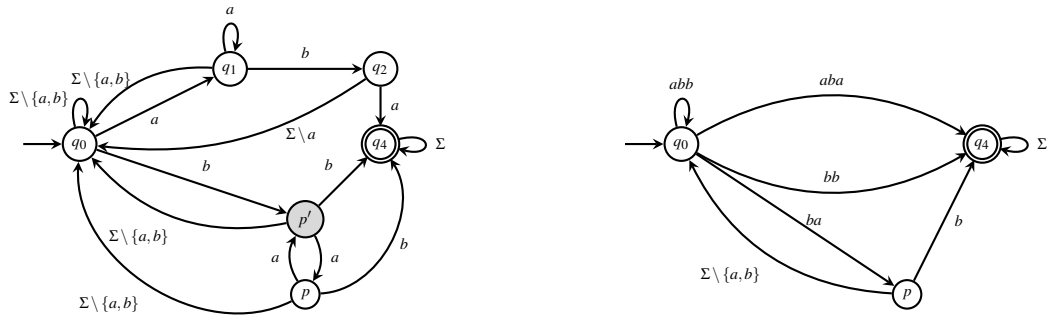
\begin{figure}
  \centering
  \begin{tikzpicture}[state/.style={circle, draw, thick, inner sep =
      2pt}]
    \begin{scope}[every node/.style={state}]
      \node (0) at (0,0) {\tiny $q_0$}; \node (1) at (2,1) {\tiny
        $q_1$}; \node (2) at (4,1) {\tiny $q_2$}; \node [double] (3)
      at (4, 0) {\tiny $q_4$}; \node [fill=gray!30] (4) at (3,-1)
      {\tiny $p'$}; \node (5) at (3,-2) {\tiny $p$};
    \end{scope}
    \begin{scope}[->,>=stealth, thick, auto]
      \draw (-0.75, 0) to (0); \draw (0) to [loop above] node [left]
      {\tiny $\Sigma \setminus \{a, b\}$} (0); \draw (0) to node
      [below] {\tiny $a$} (1); \draw (1) to [bend right=30] node
      [left] {\tiny $\Sigma \setminus \{a, b\}$} (0); \draw (1) to
      [loop above] node {\tiny $a$} (1); \draw (1) to node {\tiny $b$}
      (2); \draw (2) to [bend left=20] node [below, near start] {\tiny
        $\Sigma \setminus a$} (0); \draw (2) to node {\tiny $a$} (3);
      \draw (0) to node {\tiny $b$} (4); \draw (4) to [bend left] node
      {\tiny $a$} (5); \draw (5) to [bend left] node {\tiny $a$} (4);
      \draw (4) to node {\tiny $b$} (3); \draw (5) to [bend
      right=60]node [right] {\tiny $b$} (3); \draw (4)
      to 
      [bend left] node {\tiny $\Sigma \setminus \{a,b\}$} (0); \draw
      (5) to [bend left=60] node {\tiny $\Sigma \setminus \{a,b\}$}
      (0); \draw (3) to [loop right] node {\tiny $\Sigma$} (3);
    \end{scope}

    \begin{scope}[xshift=8cm]
      \begin{scope}[every node/.style={state}]
        \node (0) at (0,0) {\tiny $q_0$}; \node [double] (3) at (4, 0)
        {\tiny $q_4$}; \node (5) at (3,-2) {\tiny $p$};
      \end{scope}
      \begin{scope}[->, >=stealth, thick, auto]
        \draw (-0.75, 0) to (0); \draw (0) to [bend left] node {\tiny
          $aba$} (3); \draw (0) to [bend right] node {\tiny $bb$} (3);
        \draw (0) to node [right] {\tiny $ba$} (5); \draw (5) to node
        [right] {\tiny $
          b$} (3); \draw (5) to [bend left] node {\tiny
          $\Sigma\setminus \{ a, b\}$} (0); \draw (0) to [loop above]
        node {\tiny $abb$} (0); \draw (3) to [loop right] node {\tiny
          $\Sigma$} (3);
      \end{scope}
    \end{scope}
  \end{tikzpicture}
  \caption{DFA $M^{**}$ on the left and a derived DSA $\Aa^{**}_S$
    with $S = \{q_0, q_2, q_4, p\}$ on the right.}
  \label{fig:m-star-star}
\end{figure}

Figure~\ref{fig:dfa-m-star-dsa} shows a DFA $M^*$. Observe that $M^*$
is minimal: every pair of states has a distinguishing suffix. Let us
now look at DSAs that can be derived from $M^*$. Firstly, any
suffix-tracking set on $M^*$ would contain $q_0, q_4$ (since they are
initial and accepting states). If $p$ is not picked, the transition
$p \xra{a} p$ is not suffix-compatible. Therefore, $p$ should belong
to the selected set. If $p$ is picked, and $q_2$ not picked, then the
set is not well-formed (see Definition~\ref{def:well-formed-set}): the
simple word $b$ from $q_0$ to $p$ is a suffix of the simple word $ab$
to $q_2$. Therefore, any suffix-tracking set should contain the $4$
states $q_0, p, q_2, q_4$. This set $S = \{q_0, p, q_2, q_4\}$ is
indeed suffix-tracking, and the DSA derived using $S$ is shown in the
right of Figure~\ref{fig:dfa-m-star-dsa}. The only other
suffix-tracking set is the set $S'$ of all states. The DSA derived
using $S'$ will have state $q_1$ in addition, and the transitions
$\Sigma \setminus \{a, b\}$. If $\Sigma$ is sufficiently large, this
DSA would have total size bigger than $\Aa^*_S$. We deduce $\Aa^*_S$
to be the smallest DSA that can be derived from $M^*$.

Figure~\ref{fig:m-star-star} shows DFA $M^{**}$ which is obtained from
$M^*$ by duplicating state $p$ to create a new state $p'$, which is
equivalent to $p$. So $M^{**}$ is language equivalent to $M^*$, but it
is not minimal. Here, if we choose $p$ in a suffix-tracking set, the
simple word to $p$ is $ba$, which is not a suffix of $ab$ (the simple
word to $q_2$). Hence, we are not required to add $q_2$ into the
set. Notice that $S = \{q_0, p, q_4\}$ is indeed a suffix-tracking set
in $M^{**}$. The derived DSA $\Aa^{**}_S$ is shown in the right of the
figure. The ``heavy'' transition on $\Sigma \setminus a$
disappears. There are some extra transition, like $q_0 \xra{bb} q_4$,
but if $\Sigma$ is large enough, the size of $\Aa^{**}_S$ will be
smaller than $\Aa^*_S$. This shows that starting from a big DFA helps
deriving a smaller DSA, and in particular, the canonical DFA of a
regular language may not derive a minimal DSA for the language.


\section{Complexity of minimization}
\label{sec:complexity}

The goal of this section is to prove the following theorem.

\begin{theorem}
  \label{thm:np-complete}
  Given a DFA $M$ and positive integer $k$, deciding whether there
  exists an equivalent DSA of total size $\le k$ equivalent to $M$ is
  NP-complete.
\end{theorem}

If $k$ is bigger than the size of the DFA $M$, then the answer is
trivial. Therefore, let us assume that $k$ is smaller than the DFA
size. For the $\NP$ upper bound, we guess a DSA of total size $k$,
compute its tracking DFA in time $\mathcal{O}(k \cdot |\Sigma|)$ and
check for its language equivalence with the given DFA $M$. This can be
done in polynomial-time by minimizing both the DFA and checking for
isomorphism.

The rest of the section is devoted to proving the lower bound.  We
provide a reduction from the minimum vertex cover problem which is a
well-known $\NP$-complete problem~\cite{DBLP:conf/coco/Karp72}. A
vertex cover of an undirected graph $G = (V, E)$ is a subset
$S \subseteq V$ of vertices, such that for every edge $e \in E$, at
least one of its end points is in $S$. The decision problem takes a
graph $G$ and a number $k' \ge 1$ as input and asks whether there is a
vertex cover of $G$ with size at most $k'$. Using the graph $G$, we
will construct a DFA $M_G$ over an alphabet $\Sigma_G$. We then show
that $G$ has a vertex cover of size $\le k'$ iff $M_G$ has an
equivalent DSA with total size $\le k$ where
$k = (k'+2)\times 2\Delta + (2 \Delta - 1)$. Here, $\Delta$ is a
sufficiently large polynomial in $|V|, |E|$ which we will explain
later. 

  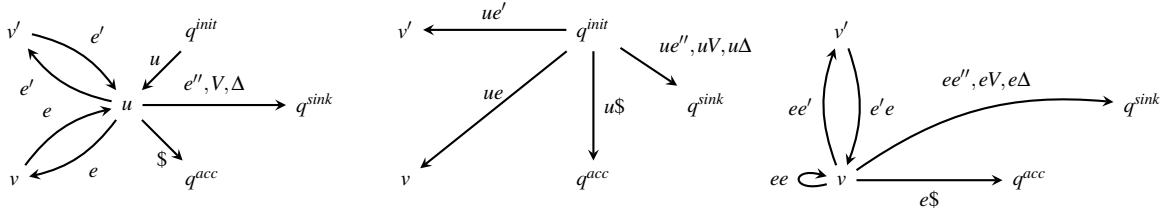
\begin{figure}
    \centering
    \begin{tikzpicture}
      \begin{scope}
        \node (u) at (0,0) {\scriptsize $u$}; \node (v) at (-1.5, -1)
        {\scriptsize $v$}; \node (init) at (1, 1) {\scriptsize
          $q^{init}$}; \node (acc) at (1, -1) {\scriptsize $q^{acc}$};
        \node (sink) at (2.5, 0) {\scriptsize $q^{sink}$}; \node (v')
        at (-1.5, 1) {\scriptsize $v'$};
      \end{scope}
      \begin{scope}[->,>=stealth, thick, auto]
        \draw (v) to [bend left = 20] node {\scriptsize $e$} (u);
        \draw (u) to [bend left = 20] node {\scriptsize $e$} (v);
        \draw (init) to node [left, near start] {\scriptsize $u$} (u);
        \draw (u) to node {\scriptsize $e'', V, \Delta$} (sink); \draw
        (u) to node [below] {\scriptsize $\$$} (acc);

        \draw (v') to [bend left=20] node {\scriptsize
          $e'$} (u); \draw (u) to [bend left=20] node [near end]
        {\scriptsize $e'$} (v');
      \end{scope}

      \begin{scope}[xshift=5.2cm]
        \begin{scope}
          \node (v) at (-1.5, -1) {\scriptsize $v$}; \node (init) at
          (1, 1) {\scriptsize $q^{init}$}; \node (acc) at (1, -1)
          {\scriptsize $q^{acc}$}; \node (sink) at (2.5, 0)
          {\scriptsize $q^{sink}$}; \node (v') at (-1.5, 1)
          {\scriptsize $v'$};
        \end{scope}
       
      \begin{scope}[->,>=stealth, thick, auto]
        \draw (init) to node [above] {\scriptsize $ue$} (v); \draw
        (init) to node [above] {\scriptsize $ue'$} (v'); \draw (init)
        to node {\scriptsize $u \$$} (acc); \draw (init) to node
        {\scriptsize $ue'', uV, u\Delta$} (sink);
      \end{scope}
    \end{scope}

    \begin{scope}[xshift=11cm]
      \node (v) at (-1.5, -1) {\scriptsize $v$};
      \node (acc) at (1, -1) {\scriptsize
        $q^{acc}$}; \node (sink) at (2.5, 0) {\scriptsize
        $q^{sink}$}; \node (v') at (-1.5, 1) {\scriptsize $v'$};
    \end{scope}

      \begin{scope}[->,>=stealth, thick, auto]
        \draw (v) to [bend left = 20] node {\scriptsize $ee'$} (v');
        \draw (v') to [bend left=20] node {\scriptsize $e'e$} (v);
        \draw (v) to node [below] {\scriptsize $e\$$} (acc); \draw (v)
        to [bend left=20] node [near end] {\scriptsize $ee'', eV,
          e\Delta$} (sink); \draw (v) to [loop left] node {\scriptsize
          $ee$} (v);
      \end{scope}
    \end{tikzpicture}
    \caption{Left: Illustration of the neighbourhood of state
      $u$ in the DFA
      $M_G$. Middle, Right: Transitions induced from
      $q^{init}$ and $v$, on removing $u$.}
    \label{fig:complexity-mg}
  \end{figure}

  The alphabet $\Sigma_G$ is given by $V \cup E \cup \{ \$\} \cup
  D$ where $D = \{1,2, \dots, \Delta\}$. States of $M_G$ are $V \cup
  \{q^{init}, q^{sink},
  q^{acc}\}$. For simplicity, we use the same notation for
  $v$ as a vertex in $G$, $v$ as a letter in $\Sigma_G$ and
  $v$ as a state of $M_G$. The actual role of
  $v$ will be clear from the context. For every edge $e = (u,
  v)$, there are two transitions in the automaton: $u \xra{e}
  v$ and $v \xra{e} u$. For every $v \in
  V$, there are transitions $q^{init} \xra{v} v$ and $v \xra{\$}
  q^{acc}$. This automaton can be completed by adding all missing
  transitions to the sink state
  $q^{sink}$. Figure~\ref{fig:complexity-mg} (left) illustrates the
  neighbourhood of a state $u$. The notation
  $e''$ stands for any edge that is not incident on
  $u$; there is one transition for every such
  $e''$. Initial and accepting states are respectively
  $q^{init}$ and $q^{acc}$.
  Let
  $L_G(u)$ be the set of words that have an accepting run in
  $M_G$ starting from
  $u$ as the initial state. 
  If $u \neq v$, $L_G(u) = L_G(v)$ implies
  $(u,v)$ is an edge and there are no other edges outgoing either from
  $u$ or
  $v$. To avoid this corner case, we restrict the vertex cover problem
  to connected graphs of 3 or more vertices. Then we have
  $M_G$ to be a minimal DFA, with no two states equivalent. Here are
  two main ideas.

  \emph{Suppressing a state.} Suppose state $u$ of $M_G$
  is suppressed (i.e. $u$ is not in a suffix-tracking
  set). In Figure~\ref{fig:complexity-mg}, we show
  the induced transitions from $q^{init}$ and a vertex $v$. However,
  some of them will be useless transitions: most importantly, the set
  of transitions $q^{init} \xra{u1, u2, \dots, u\Delta} q^{sink}$ will
  be useless bigger-suffix-transitions due to
  $q^{init} \xra{1, 2, \dots, \Delta} q^{sink}$. Similarly,
  $v \xra{e1, e2, \dots, e\Delta} q^{sink}$ will be removed. There are
  some more useless bigger-suffix-transitions, like
  $v \xra{e e''} q^{sink}$ for some $e''$ that is not incident on $v$
  and $u$. So from each $v$, at most $2 |E|$ transitions are
  added. But crucially, after removing
  useless transitions, the $\Delta$ transitions from $u$
  no longer appear. If we choose $\Delta$ large enough to compensate
  for the other transitions, we get an overall reduction in size by
  suppressing states.

  \emph{Two states connected by an edge cannot both be suppressed.}  Suppose $e = (u, v)$ is an edge. If $S$ is a set where
  $u, v \notin S$, then the transition $v \xra{e} u$ is not
  suffix-compatible: the simple word $ue$ from $q^{init}$ to $v$, when
  extended with $e$ gives the word $uee$; no suffix of $uee$ is a
  simple word from $q^{init}$ to $u$. We deduce that
  suffix-tracking sets in $M_G$ correspond to a vertex cover in $G$,
  and vice-versa.

  These two observations lead to a translation from minimum vertex
  cover to suffix-tracking sets with least number of states. Due to
  our choice of $\Delta$, DSAs with smallest (total) size are indeed
  obtained from suffix-tracking sets with the least number of states. 
Let $k = (k' + 2) \times 2 \Delta + (2 \Delta - 1)$.

\subparagraph*{Vertex cover $\le k'$ implies DSA $\le k$.}
Assume there is a vertex cover $\{v_1,
  \dots, v_p \}$ in $G$ with $p \le k'$. Let $S$ be the set of states in $M_G$ corresponding to $\{v_1,
  \dots, v_p \}$. Observe that $S \cup \{q^{init},q^{sink},q^{acc}\}$ is a suffix-tracking set; every transition is trivially suffix-compatible ($\forall q\xra{a}u, q\in S \text{ or } u\in S $). Well-formedness holds because $\forall p,q \in S, \a \in \spath{p}{q}{S}$ we have $|\a| \le 2$; this means $\forall q' \notin S, \b \in \spath{p}{q'}{S}$, we have
  $\a \not \sfx \b$ (since $|\b|=1$). Hence the derived DSA will be equivalent to $M$.

  The derived DSA has $p + 3$ states, and transitions $q \xra{1, 2, \dots,
\Delta} q^{sink}$ from each except for the $q^{sink}$ state. The
transitions on $q^{sink}$ are removable, and hence will be absent. All of this adds $(p+2) \times 2 \Delta$ to the total size (edges +
label lengths). Apart from these, there are transitions with
labels of length at most $2$, over the alphabet $V \cup E \cup \$$. From each vertex, $v$, there are $|V|$ transitions to $q^{sink}$, one transition to $q^{acc}$ and at most $2|E|$ transitions to other states or $q^{sink}$. We
can choose a large enough $\Delta$ (say $(|V| + |E|)^4$), so that the size of
these extra transitions is at most $2 \Delta - 1$. 
  Hence, total size is $\le (p + 2) \times 2 \Delta + (2 \Delta - 1)$.

By assumption, we have $p \le k'$. Therefore, the size of the DSA 
is $\le (k' + 2) \times 2 \Delta + (2 \Delta - 1)
=k$.

  \subparagraph*{DSA $\le k$ implies vertex cover $\le k'$.}

  Let $\Aa$ be a DSA with size $\le k$. It may not be derived from $M_G$. However, by Proposition~\ref{prop:dsa-derivable-from-tracking-dfa} we know $\Aa$ is derived from a DFA $M$, the tracking DFA for $\Aa$. Moreover since $M_G$ is the minimal DFA, we know that $M$ will be a \emph{refinement} of $M_G$ (see Section~\ref{sec:preliminaries} for definition).

  Let us consider a pair of states $u$ and $v$ from $M_G$, such that the vertices $u,v \in G$ have an edge between them labeled $e$. The DFA $M$ will have two sets of states $u_1,u_2,\dots,u_i$ and $v_1,v_2,\dots,v_j$ that are language-equivalent to $u$ and $v$ respectively. Its initial state must have a transition on $v$ to one of $v_1,v_2,\dots,v_j$. Without loss of generality, let it be to $v_1$. Each of $v_1,v_2,\dots,v_j$ must have a transition on $e$ to one of $u_1,u_2,\dots,u_i$ (for equivalence with $M_G$) and vice-versa. Consider the run from the initial state on $ve^{i+j+1}$. At least one of the states among $u_1,u_2,\dots,u_i, v_1,v_2,\dots,v_j$ must be visited twice; consider the first such instance. The transition on $e$ that re-visits a state cannot be suffix-compatible w.r.t a set $S$, if none of these states are in $S$. For it to be suffix-compatible, the string $ve^k.e$ (from initial state to the first repeated state) must have its longest simple-word suffix go the same state. Since $ve^k.e$ is not simple by itself, its longest suffix must consist entirely of $e$'s. But on any string of $e$'s, the initial state moves only to the sink state(s) and not to any of $u_1,u_2,\dots,u_i, v_1,v_2,\dots,v_j$. Hence any suffix-tracking set must contain at least one of these states, which maps to at least one of $u$ or $v$ in $G$. Every suffix-tracking set of $M$ therefore maps to a vertex cover $\{v_1, v_2, \dots, v_p\}$.

  Now we show that the size of this vertex cover is $\le k'$.
Each of the states picked in the suffix-tracking set will contribute to atleast $2 \Delta$ in the total size, due to the $\Delta$ transitions. We will also have these $\Delta$ transitions from the initial and accepting states. Therefore, the total size is $(p + 2) \times 2 \Delta + y$ for some $y > 0$. Hence $(p + 2) \times 2 \Delta \le k$. This implies $p \le k'$: otherwise we will have $p \ge k'+1$, and hence $(p + 2) \times 2 \Delta \ge (k' + 1 + 2) \times 2 \Delta = (k' + 2) \times 2 \Delta + 2 \Delta > k$, a contradiction.


\section{Conclusion}
\label{sec:conclusion}

We have introduced the model of deterministic suffix-reading automata,
compared its size with DFAs and DGAs, proposed a method to derive DSAs
from DFAs, and presented the complexity of minimization. 
The work on DGAs~\cite{giammarresi1999deterministic} inspired us to look for methods to derive DSAs from
DFAs, and investigate whether they lead to minimal DSAs for a
language. This led to our technique of suffix-tracking sets, which
derives DSAs from DFAs. The technique imposes some natural conditions
on subsets of states, for them to be tracking patterns at each
state.
However, surprisingly, the smallest DSA that we can derive from the
canonical DFA need not correspond to the minimal DSA of a
language. 
This leads to several questions about the DSA model, and our
derivation methodology.

When does the smallest DSA derived from the canonical DFA correspond
to a minimal DSA? 
Can we use our techniques to study minimality in terms of number of
states?  Closure properties of DSAs - do we perform the union,
intersection and complementation operations on DSAs without computing
the entire equivalent DFAs? What about practical studies of using
DSAs?  To sum up, we believe the DSA model offers advantages in the
specification of systems and in also studying regular languages
from a different angle. The results that we have presented throw
light on some of the different aspects in this model, and lead to
many questions both from theoretical and practical
perspectives.


\bibliographystyle{eptcs} \bibliography{dsa}

\end{document}